\documentclass[draftclsnofoot,onecolumn,11pt,conference]{IEEEtran}
\usepackage{setspace}

\usepackage{cite}
\usepackage{psfrag}
\usepackage{subfigure}
\usepackage{amsmath}
\usepackage{amssymb}

\usepackage[pdftex]{graphicx}

\newtheorem{theorem}{Theorem}

\newtheorem{lemma}[theorem]{Lemma}

\newtheorem{proposition}[theorem]{Proposition}

\newcommand{\bydef}{\triangleq}

\def\bb0{{\mathbb{0}}}

\def\bb{{\mathbf{b}}}

\def\bff{{\mathbf{f}}}
\def\bg{{\mathbf{g}}}
\def\bh{{\mathbf{h}}}

\def\bw{{\mathbf{w}}}
\def\bx{{\mathbf{x}}}
\def\by{{\mathbf{y}}}

\def\b0{{\mathbf{0}}}

\def\bI{{\mathbf{I}}}

\def\bW{{\mathbf{W}}}
\def\bX{{\mathbf{X}}}

\def\bbC{{\mathbb{C}}}

\def\bbE{{\mathbb{E}}}

\def\cK{\mathcal{K}}

\def\cN{\mathcal{N}}

\def\cQ{\mathcal{Q}}

\def\cS{\mathcal{S}}

\def\sf0{{\mathsf{0}}}

\def\Nt{{N_t}}

\def\th{{\rm th}}
\def\l{{\ell}}
\def\Bt{{B_{\rm tot}}}
\def\Bk{{B_{k}}}
\def\Bi{{B_{i}}}
\def\Bkl{{B_{k,{\ell}}}}
\def\Dk{{D_{k}}}
\def\Dkl{{D_{k,{\ell}}}}
\def\hk{{\mathbf{h}_{k}}}
\def\gkl{{\mathbf{g}_{k,{\ell}}}}
\def\skl{{\mathbf{s}_{k,{\ell}}}}
\def\wkl{{\mathbf{w}_{k,{\ell}}}}
\def\wk{{\mathbf{w}_{k}}}
\def\sk{{\mathbf{s}_{k}}}
\def\thk{{\tilde{\mathbf{h}}_{k}}}
\def\tgkl{{\tilde{\mathbf{g}}_{k,{\ell}}}}
\def\twkl{{\tilde{\mathbf{w}}_{k,{\ell}}}}

\def\hhk{{\hat{\mathbf{h}}_{k}}}
\def\fk{{\mathbf{f}_{k}}}
\def\fl{{\mathbf{f}_{\ell}}}
\def\hfk{{\hat{\mathbf{f}}_{k}}}
\def\hfl{{\hat{\mathbf{f}}_{\ell}}}
\def\hgkl{{\hat{\mathbf{g}}_{k,{\ell}}}}

\def\nth{n^{\rm th}}
\def\kth{k^{\rm th}}
\def\lth{\ell^{\rm th}}

\begin{document}

\title{Adaptive Bit Partitioning for\\Multicell Intercell Interference Nulling with\\Delayed Limited Feedback}
\author{\authorblockN{Ramya Bhagavatula and Robert W. Heath, Jr.}
\authorblockA{Department of Electrical \& Computer Engineering \\
Wireless Networking and Communications Group\\
The University of Texas at Austin\\
1 University Station, C0803 \\
Austin, TX 78712-0240 \\
ramya.bh@mail.utexas.edu, rheath@ece.utexas.edu\\}}
\maketitle 

\begin{abstract}
Base station cooperation can exploit knowledge of the users' channel state information (CSI) at the transmitters to manage co-channel interference. Users have to feedback CSI of the desired and interfering channels using finite-bandwidth backhaul links. Existing codebook designs for single-cell limited feedback can be used for multicell cooperation by partitioning the available feedback resources between the multiple channels. In this paper, a new feedback-bit allocation strategy is proposed, as a function of the delays in the communication links and received signal strengths in the downlink. Channel temporal correlation is modeled as a function of delay using the Gauss-Markov model. Closed-form expressions for bit partitions are derived to allocate more bits to quantize the stronger channels with smaller delays and fewer bits to weaker channels with larger delays, assuming random vector quantization. Cellular network simulations are used to show that the proposed algorithm yields higher sum-rates than an equal-bit allocation technique. 
\end{abstract}

\newpage

\section{Introduction} \label{sec:Intro}  
Multicell base station cooperation can improve sum-rates and reduce outage in cellular systems \cite{Shamai2001a, Jafar2004, Zhang2004,Jing2008, Fettweis2010,Simeone2009, Papadogiannis2008, Samardzija2009, Ali2009, Choi2008, Ng2005, Ng2008,Ekbal2005, Zhang2009c,Lee2008, Lee2009,Bhagavat2009,Bhagavat2010}. By sharing user channel state information (CSI) and/or data and/or precoding matrices via high-capacity wired backhaul links, base stations coordinate transmissions to manage co-channel interference. Consequently, cooperation can be used to improve the performance of systems that have high levels of co-channel interference due to universal frequency reuse, e.g., upcoming commercial wireless standards like 3GPP long term evolution (LTE) Advanced \cite{3GPPAdv}. Knowledge of the desired and interfering CSI at the base stations is important to obtain the full performance gains promised by multicell cooperation. While it is commonly assumed in the literature that full CSI is available at all base stations, the feedback channel has finite bandwidth implying that limited feedback techniques \cite{Love2008} will be employed to obtain information of \emph{multiple} channels at the transmitters \cite{Bhagavat2009}. Further, the multicell cooperative literature also commonly assumes zero-delay backhaul and feedback links. This is not realistic due to propagation and signal processing delays \cite{Zhang2009b}. Successful implementation of cooperative strategies in next generation cellular systems calls for evaluating the performance gains obtained by accounting for finite-bandwidth feedback channels and delayed CSI availability. 

Strategies that utilize full cooperation, e.g. dirty paper coding \cite{Costa1983, Shamai2001a, Zhang2004, Jafar2004}, multicell zero-forcing, minimum mean square error, and null-space decomposition \cite{Zhang2004,Jing2008}, typically yield higher sum-rates. They generally possess, however, high-complexity and involve the exchange of large amount of information between base stations, increasing the load on finite-capacity backhaul links \cite{Fettweis2010,Simeone2009, Papadogiannis2008,Samardzija2009}. In contrast, approaches that employ joint dynamic resource allocation \cite{Ali2009} and/or joint scheduling \cite{Choi2008} to ensure orthogonality among the transmissions of users in neighboring cells tradeoff performance gains in exchange for low complexity and backhaul loads. Partial base station cooperation (known as \emph{coordinated beamforming} in 3GPP LTE Advanced \cite{3GPPAdv}) is an effective compromise. It offers reasonable sum rate improvements and only a small amount of additional backhaul bandwidth for moderate Doppler spreads \cite{Samardzija2009}. Each base station designs beamforming vectors on-site to transmit exclusively to the users in its own cell by exchanging only user CSI on the backhaul. Examples of partial cooperative strategies include MMSE-estimation-based beamforming \cite{Ng2005, Ng2008}, transmit power-minimizing beamforming \cite{Ekbal2005}, intercell interference nulling (ICIN) \cite{Zhang2009c,Jorswieck2008,Lindblom2009}, and sum-rate maximizing beamforming for single-interferer scenarios \cite{Lee2008,Lee2009,Bhagavat2009,Bhagavat2010}.  In this paper, we use ICIN, which is non-iterative, does not suffer from convergence issues (such as those in  in \cite{Ng2005, Ng2008}), and is suitable for multiple interferer scenarios. 

The performance of a multicell cooperative transmission strategy is highly dependent on the quality of the CSI fed back by the users. Most of the literature on multicell cooperation assumes that full CSI is available at the transmitters \cite{Fettweis2010,Simeone2009, Papadogiannis2008, Samardzija2009, Ali2009, Choi2008, Shamai2001a, Zhang2004,Jafar2004, Kang2006,Jing2008, Ng2005, Ng2008, Somekh2006,Ekbal2005,Lee2008, Lee2009}. Limited feedback for multicell systems is a topic of ongoing research \cite{Zhang2009c,Bhagavat2009,Bhagavat2010}. Unfortunately, results from the well-investigated single-cell limited feedback \cite{Love2008} are not directly applicable to the multicell scenario. While the CSI of only one channel is fed back in the single-cell case, cooperative strategies require feedback of CSI from \emph{multiple} base stations using the same feedback link. Further, in single-cell transmission, quantized CSI reaches the base station after experiencing a delay in the feedback channel \cite{Zhang2009b, Ma2009,Isukapalli2007,Hu2008,Akoum2010}. In the multicell cooperative framework, however, quantized CSI is subject to an \emph{additional} source of delay in the backhaul link. The impact of delayed CSI on the performance of (single cell) non-cooperative systems \cite{Zhang2009b, Ma2009,Isukapalli2007,Hu2008,Akoum2010} has been investigated extensively. Note that while \cite{Akoum2010} accounts for inter-cell interference, base station cooperation is not included in the system model. The effect of delayed limited feedback on the performance of cooperative systems has received comparatively less attention.

Jointly quantizing the CSI of multiple channels was proposed as a solution to multicell limited feedback in \cite{Bhagavat2010} for a single interferer. This, however, requires design of special \emph{multi-channel} codebooks, which is a topic of ongoing investigation. Existing codebook designs can be exploited if the CSI of each of the channels is quantized independently. Hence, it is reasonable to quantize the desired and interfering CSI using separate codebooks, by partitioning the available feedback bits among the different channels \cite{Bhagavat2009, Zhang2009c}. Intuitively, the bits allocated to feedback each of the multiple channels should be a function of the signal strength and delay experienced by the channels. For example, a weak interfering channel can be allocated few or no bits since it has a small impact on the sum-rate. The bits can, instead, be allocated to quantize a stronger interferer more finely to reduce the associated quantization error and thereby, reduce its impact on the sum-rate. Similarly, channels experiencing large delays can be assigned fewer bits since nulling out the outdated channel might not contribute to improving sum-rates. The bit allocations derived in \cite{Bhagavat2009} are not applicable to multiple interferer scenarios; \cite{Zhang2009c} does not provide closed-form solutions for the bit-partitioning. Further, \cite{Bhagavat2009, Zhang2009c} assume zero delay backhaul and limited feedback links.

In this paper, we propose a feedback-bit partitioning algorithm to reduce the mean loss in sum-rate due to delayed limited feedback in a multicell cooperative system using ICIN, by accounting for both average signal strength and delay. We consider multi-antenna base stations and single antenna receivers, i.e. multiple-input single-output (MISO) systems. We assume a single active user in each cell, which faces interference from several neighboring cells. We model channel temporal correlation using the Gauss-Markov autoregressive model \cite{Clarke1968,Zhang2009b, Ma2009,Isukapalli2007,Hu2008,Akoum2010} and assume that each user estimates perfectly and quantizes the channels using random vector quantization (RVQ) \cite{Jindal2006}, for analytical reasons. The quantized CSI is fed back to its own base station, which then exchanges this information with adjacent base stations over the backhaul links. Thus, each base station has knowledge of its desired channel and the interference it \emph{causes} to the neighboring users. We first quantify the impact of delayed limited feedback on the performance of the cooperative ICIN strategy by deriving an upper bound on the mean loss in sum-rate. We then develop closed-form expressions to determine feedback-bit allocations to reduce the mean loss in sum-rate. Simulations verify the performance of the proposed algorithm. 

The contributions of the paper are summarized as follows.
\begin{itemize}
 \item{We derive an upper bound on the mean loss in sum-rate due to delayed limited feedback in a MISO system using ICIN, assuming a Gauss-Markov autoregressive channel model and RVQ. The bound is a function of the number of feedback bits allocated to each of the channels, the relative strengths, and delays of the channels.}
 \item{We derive closed-form expressions for computing the feedback bit partitions among the desired and interfering channels, as a function of their relative strengths and delays.}
  \item{We use simulations to verify that the proposed feedback-bit allocation algorithm increases the sum-rate per base station using cooperative ICIN with delayed limited feedback. }
\end{itemize}
 
The paper is organized as follows. The system model is presented in Section~\ref{sec:SysModel}. Cooperative ICIN, assuming full  and limited CSI is described in Section~\ref{sec:ICIN}. The impact of delayed limited feedback on the mean sum-rate using ICIN is presented in Section~\ref{sec:Delay_BD}. The feedback bit-partitioning algorithm is proposed in Section~\ref {sec:FB_BitPart} to reduce the mean loss in sum-rate. In Section~\ref{sec:Results}, simulations are presented to verify that the mean sum-rate can be improved using the feedback technique in this paper. Finally, the conclusions are provided in Section~\ref{sec:Conclusion}.

\noindent\textbf{Notation:} In this paper, $\bX$ refers to a matrix and $\bx$ to a vector. $\bX^T$ and $\bX^*$ denote the transpose and Hermitian transpose, respectively. The pseudo-inverse of $\bX$ is given by $\bX^\dag$. An identity matrix of size $R \times R$ is denoted by $\bI_R$. $\bbE\{.\}$ refers to the expectation. $\|\bx\|$ stands for the Frobenius norm of $\bx$. $\cN_c(\mu,\sigma^2)$ refers to a complex Gaussian distribution with mean $\mu$ and variance $\sigma^2$. The $n^{\rm th}$ column $\bX$ is denoted by $\bX(:,n)$. The angle between two vectors, $\bx$ and $\by$ is denoted by $\theta_{(\bx,\by)}$. The cardinality of a set $\cS$ is denoted by $|\cS|$.

\section{System Model} \label{sec:SysModel}
Consider the cellular system in Fig.~\ref{fig:SystemModel}. We assume that the base station in each cell serves a single active user, using intra-cell time division multiple access (TDMA) or a comparable orthogonal access strategy and that the user faces interference from the neighboring cells \cite{Shamai2001a}. The strength of the interfering signals at each user depends on the  respective propagation channels and the distance between the interfering base station and user. This is similar to the approaches used in \cite{Jing2008,Bhagavat2009,Bhagavat2010}. The number of cells in the multicell system is denoted by $K$. We index the users in each cell by the base station they obtain their desired signal from, i.e. the $k^{\rm th}$ base station services the $k^{\rm th}$ user, for $k=1,\ldots,K$. We assume that all the base stations are equipped with $\Nt$ antennas, while each user supports a single receive antenna. The channel corresponding to the desired signal between the $k^{\rm th}$ base station and $k^{\rm th}$ user is denoted by $\bh_k[n]\in \bbC^{\Nt \times 1}$, at the $\nth$ discrete-time instant. The interfering channel between the $k^{\rm th}$ user and the $\l^{\rm th}$ base station is given by $\bg_{k,\l}[n] \in \bbC^{\Nt \times 1}$ for $k\neq\l$. This is illustrated in Fig.~\ref{fig:SystemModel}. 

Note that the results derived in the paper can be extended to receivers with multiple antennas using multiuser eigenmode transmission, a block-diagonalization based technique, instead of ICIN to null out the interference between multiple users \cite{Boccardi2007}. In this paper, we model the desired and interfering channels by the i.i.d. Rayleigh fading model,  where each entry is unit variance complex Gaussian independent and identically distributed (i.i.d.) according to $\cN_c(0,1)$. While it is recognized that the Rayleigh fading model does not model realistic propagation channels accurately, we use the i.i.d. Gaussian assumption to facilitate the limited feedback analysis. 

\subsection{Discrete-Time Input-Output Model}\label{sec:LFModel}
The symbol transmitted from the $k^{\rm th}$ base station (intended for the $k^{\rm th}$ user) is denoted by $s_k$, where $\bbE\{|s_k|^2\} = E_{s}$. Each user is assumed to face interference from $K-1$ neighboring base stations, each transmitting with energy $E_s$. The path-loss incurred by the desired signal is given by $L_k$, and that by the interfering signal from the $\lth$ base station to the $\kth$ user is given by $L_{k,\l}$. The received signal powers of the desired and interfering signals are then given by $\gamma_k = E_{s}/L_k$ and $\gamma_{k,\l} = E_{s}/L_{k,\l}$, respectively. We define the interference to signal noise ratio (ISR) of the $\l^{\rm th}$ interferer to the $k^{\rm th}$ user by $\alpha_{k,\l} = \gamma_{k,\l}/\gamma_{k}$, where $\l = 1,\ldots,K$, $\l\neq k$. Further, $\alpha_{k,\l} \in [0,1]$ (i.e. the interfering signal strength can at most be equal to that of the desired signal, otherwise the user will associate with a different base station). Note that a similar parameter is used in \cite{Jing2008,Bhagavat2009,Bhagavat2010} to model the received power of the interfering signal with respect to the desired signal. Using the narrowband flat-fading model, the baseband discrete-time input-output relation for the user in the $k^{\rm th}$ cell is given by
\begin{equation}
y_k[n] = \sqrt{\gamma_k} \bh_{k}^*[n] \bff_{k}[n] s_{k}[n] + \sum_{\substack{\l = 1 \\ \l \neq k}}^{K }\sqrt{\alpha_{k,\l}\gamma_k}\bg_{k,\l}^*[n]\bff_{\l}[n] s_{\l}[n] + v_k[n] ,\label{eqn:InpOutReln}
\end{equation}
where $y_k[n] \in \bbC$ is the received signal at the $k^{\rm th}$ user and $\bff_k[n]\in \bbC^{\Nt \times 1}$ is the unit-norm beamforming vector at the $k^{\rm th}$ base station, at the $n^\th$ time instant. Finally, $v_k[n] \in \bbC$ is complex additive zero-mean white Gaussian noise at the $k^{\rm th}$ user, with $\bbE\{|v_k|^2\} = N_o$. We denote the received desired and interfering signal to noise ratios (SNR) as $\rho_k = \gamma_k/N_o$ and $\alpha_{k,\l}\rho_k$, respectively. The base stations are assumed to have perfect knowledge of $\rho_{k}$. This is a popular assumption in literature \cite{LF-Assumption1,LF-Assumption2,Huang2009}. It was shown in \cite{Huang2009} that SNR quantization does not effect the sum-rates of a single-cell multiuser MIMO system signiÞcantly. To the best of our knowledge, the effect of SNR quantization on the sum-rates of a multicell system has not yet been investigated. 

The signal to interference noise ratio (SINR) of the $k^{\rm th}$ user at the $n^\th$ instant is given by
\begin{equation}
\mathsf {SINR}_k[n] = \frac{\rho_{k}|\bh^*_{k}[n]\bff_{k}[n]|^2}{1 + \sum_{\substack{\l = 1 \\ \l \neq k}}^{K }\alpha_{k,\l}\rho_k|\bg_{k,\l}^*[n]\bff_{\l}[n]|^2}
\label{eqn:SINR1}
\end{equation}
The sum-rate of all the users within the system is 
\begin{equation}
R_s[n] = \sum_k \log_2\left( 1 + \mathsf{SINR}_k[n] \right).
\label{eqn:SumRate}
\end{equation}
The sum-rate, hence, depends on beamforming vectors $\bff_k[n]$, which are designed using quantized channel state information. Note that in the remainder of this paper, we assume that $k\neq\l$ unless otherwise mentioned.

\subsection{Limited Feedback Model}\label{sec:LFModel}

The channel directions, denoted by $\tilde{\bh}_k[n] \bydef \bh_k[n]/\|\bh_k[n]\|$ and $\tilde{\bg}_{k,\l}[n] \bydef \bg_{k,\l}[n]/\|\bg_{k,\l}[n]\|$, are quantized to the unit-norm vectors given by $\hat{\bh}_k[n]$ and $\hat{\bg}_{k,\l}[n]$, respectively, at the $k^\th$ user in the $n^\th$ time instant. Using ICIN, the cooperative strategy used in this paper, beamforming vectors are designed to lie in the null space of the interfering channel \emph{directions}. Since base stations do not require knowledge of the channel gains, users feedback only the estimated channel directions. 

Channel directions can be quantized either jointly using a single codebook \cite{Bhagavat2010} or independently using separate codebooks \cite{Bhagavat2009, Zhang2009c}. It was proposed in \cite{Bhagavat2010} that desired and interfering channel directions can be jointly quantized by introducing an `inherent bias' based on the ISR. Special codebooks are required that are a function of the interfering signal strength, implying that the number of codebooks that need to be designed and stored for joint quantization can be prohibitively large. In contrast, channel directions are quantized separately in independent quantization. Thus, existing single-user codebook designs can be used \cite{Bhagavat2009, Zhang2009c}. The number of bits assigned to each channel direction is varied depending on the respective ISR. Hence, the number of codebooks required only depends on the total number of feedback bits available, $\Bt$ and consequently, is not as large as that required in the joint quantization case. Due to these reasons, we consider separate quantization in this paper, as illustrated in Fig.~\ref{fig:LFModel}.

We assume that each user can utilize $\Bt$ bits for feedback, and that $B_k$ and $B_{k,\l}$ bits are used to quantize $\tilde{\bh}_k[n]$ and $\tilde{\bg}_{k,\l}[n]$ respectively, where $\Bk + \sum_{\substack{\l=1\\ \l\neq k}}^{K} \Bkl = \Bt$. This approach of splitting available feedback bits between the desired and interfering channels was used in \cite{Bhagavat2009,Zhang2009c}. The delay associated with quantizing $\thk[n]$ (to $\hhk[n]$) and feeding back $\hhk[n]$ to the $k^\th$ base station is denoted by $\Dk$. The $k^\th$ user also quantizes the interfering channels, $\gkl[n]$ to $\hgkl[n]$ and feeds back $\hgkl[n]$ to the $k^\th$ base station, which then transmits $\hgkl[n]$ to the $\l^\th$ base station over the backhaul link, incurring a delay of $\Dkl$, where $\l = 1,\ldots,K$. Hence, at the time instant $n$, the $\kth$ base station has knowledge of $\hk[n-\Dk]$ and $\gkl[n-D{\l,k}]$, for all $\l\neq k$.  

In this paper, we assume that the delays, $\Dk$ and $\Dkl$, are known at the base station and users. This is reasonable since delays in the feedback and backhaul link can be measured using training signals, assuming perfect synchronization between the transmitter(s) and receiver(s). Note that $\Dkl\geq\Dk$, since $\Dkl$ includes exchange of CSI over the backhaul link, in addition to limited feedback from the $k^\th$ user. 

\subsection{Gauss-Markov Model for CSI Delay}
The presence of delay in the system leads to a loss in the sum-rate, which can be evaluated by understanding the relation between the current and delayed CSI. In this paper, the Gauss-Markov block fading autoregressive model is used to characterize the relation between $\bh_k[n]$ ($\gkl[n]$) and $\bh_k[n-\Dk]$ ($\gkl[n-\Dkl]$). It has been shown in the literature \cite{HaykinBook, Turin2001,Tan2000} that the Gauss-Markov autoregressive model is reasonably accurate for relatively small delays in the communication links. Hence, it has been widely used in research to model the effect of delay on the performance of wireless systems using limited feedback \cite{Zhang2009b, Ma2009,Isukapalli2007,Hu2008}. By assuming that $\bh_k[n]$ and $\gkl[n]$ are constant throughout the codeword transmission, the current and delayed CSI are related by 
\begin{eqnarray}
\hk[n] &=& \eta_k\hk[n-D_k] + \sqrt{1-\eta_k^2}\bw_\hk[n], ~~{\rm and} \label{eqn:GM_hk} \\
\gkl[n] & = & \eta_{k,\l}\gkl[n-\Dkl] + \sqrt{1-\eta_{k,\l}^2}\bw_\gkl[n], \label{eqn:GM_gkl}
\end{eqnarray}
where $\bw_\hk[n]$ and $\bw_\gkl[n]$ denote the channel error vectors, and are uncorrelated with $\hk[n-D_k]$ and $\gkl[n-\Dkl]$, respectively. The entries of $\bw_\hk[n]$ and $\bw_\gkl[n]$ are distributed by $\cN_c(0,1)$. The correlation coefficients for the desired and interfering channels are denoted by $\eta_k$ and $\eta_{k,\l}$, respectively. Clarke's autocorrelation model is used to determine $\eta_k$ and $\eta_{k,\l}$ as \cite{Clarke1968, Zhang2009b,Ma2009,Isukapalli2007,Hu2008}
\begin{eqnarray}
\eta_k &=& J_0(2\pi\Dk f_dT_s), ~~{\rm and} \label{eqn:GM_hk} \\
\eta_{k,\l} & = & J_0(2\pi\Dkl f_dT_s), \label{eqn:GM_gkl}
\end{eqnarray}
where $J_0$ is the zeroth order Bessel function of the first kind, $f_d$ is the Doppler spread and $T_s$ is the symbol duration. The Doppler spread, $f_d = \nu f_c/c$, where $\nu$ is the relative velocity of the transmitter-receiver pair, $f_c$ the carrier frequency, and $c$ the speed of light. Note that Clarke's model requires isotropic scattering, which is satisfied in this paper through the Rayleigh fading assumption.


\section{Inter-Cell Interference Nulling} \label{sec:ICIN}
In this section, we briefly describe ICIN for the setup in Fig.~\ref{fig:SystemModel}. We consider first the ideal case where full CSI is available at all the base stations and the delay associated with feedback and backhaul is zero. We then proceed to the comparatively more realistic limited feedback scenario with non-zero delay.

\subsection{Full CSI and Zero Delay} \label{sec:BD_fullCSI}
The $k^\th$ base station has instantaneous knowledge of not only its own desired channel, $\hk[n]$, but also of the interference caused to neighboring cells, i.e. $\bg_{\l,k}[n], \l = 1,\dots,K,\l \neq k$, made available via the backhaul link. The $\kth$ base station then computes the beamforming vector, $\fk,[n]$, as \cite{Zhang2009c,Jorswieck2008,Lindblom2009}
\begin{equation} 
\fk[n] = \bW_k[n](:,1),~~\rm{where}~\bW_\l[n] = \left(\left[\tilde{\bh}_k[n]~\tilde{\bg}_{1,k}[n]~\ldots~\tilde{\bg}_{K,k}[n]\right]\right)^\dag. \label{eqn:BFVec_FullCSI_ZeroD}
\end{equation}
Since all the channels in $\bW_k[n]$ are independent of each other with high probability, \eqref{eqn:BFVec_FullCSI_ZeroD} ensures perfect interference nulling, i.e. $\bg_{k,\l}^T[n]\fl[n] = 0$, for $\l = 1,\ldots,K,\l\neq k$, when $\Nt \geq K$. The sum-rate, assuming full CSI and zero delay, is given by
\begin{equation}
R_s[n] = \sum_{k=1}^K \log_2\left(1 + |\bh^T_{k}[n]\bff_{k}[n]|^2 \right), 
\label{eqn:Rs_fullCSI_ZeroD}
\end{equation}
where $\fk[n]$ is given by \eqref{eqn:BFVec_FullCSI_ZeroD}. The denominator is nulled out since the beamforming vector at the $\l^\th$ base station is designed to be lie in the null-space of the $\gkl$.  

\subsection{Limited Feedback with Delay} \label{sec:BD_LF}
As a result of delay, at the $n^\th$ time instant, the $k^\th$ base station has knowledge of its desired channel, $\hhk[n-\Dk]$ and the interference that it causes to the neighboring $K-1$ cells, $\hat{\bg}_{\l,k}[n-\Dkl]$ (for $\l=1,\ldots,K, \l\neq k$). Hence, the beamforming vector at the $n^\th$ time instant, $\hfk[n]$, is designed using the delayed and quantized CSI of the desired channels and the interference caused to other cells 
\begin{equation} \label{eqn:BFVec_LFD}
\hfk[n] = \hat{\bW}_k[n](:,1),~~\rm{where}~\hat{\bW}_k[n] = \left(\left[\hhk[n-\Dk]~\hat{\bg}_{1,k}[n- D_{1,\l}]~\ldots~ \hat{\bg}_{K,k}[n-D_{K,k}]\right]\right)^\dag .
\end{equation}
When $\Nt \geq K$, the beamforming vector lies in the $\Nt - (K-1)$ null-space of the $K-1$ interfering channels. Hence, when $\Nt = K$, $\hfk[n]$ will lie in a one-dimensional subspace, independent of $\hhk$. This implies that if $\Nt = K$, it is not necessary to feedback the quantized desired channel back to the base station, i.e. $\Bk = 0$. In contrast, when $\Nt > K$, $\hhk$ is desirable to determine the best $\hfk[n]$ in the $\Nt - (K-1)$ subspace. 

The sum-rate, assuming limited feedback and delay, is given by
\begin{equation}
\hat{R}_s[n] = \sum_{k=1}^K \log_2\left(1 + \frac{\rho_{k}|\bh^T_{k}[n]\hfk[n]|^2}{1 + \sum_{\substack{\l = 1 \\ \l \neq k}}^{K }\alpha_{k,\l}\rho_k|\bg_{k,\l}^T[n]\hfl[n]|^2} \right), 
\label{eqn:Rs_LFD}
\end{equation}
where $\hfk[n]$ and $\hfl[n]$ are given by \eqref{eqn:BFVec_LFD}. Due to limited feedback, interfering signals in the denominator of \eqref{eqn:SINR1} are not nulled out, i.e. $|\bg_{k,\l}^T[n]\hfl[n]|^2 \neq 0$. Note, however, that $|\hgkl^T[n-\Dkl]\hfl[n]|^2 = 0$. Hence, the non-zero denominator of \eqref{eqn:Rs_LFD} reduces the sum-rate of the overall system. In the next section, we quantify the mean loss in sum-rate caused by delayed limited feedback.


\section{Impact of Limited Feedback and Delay on Sum-Rate of ICIN} \label{sec:Delay_BD}
We define the mean loss in sum-rate due to delay and limited feedback, using ICIN as
\begin{equation}
\bbE\{\Delta R_s[n]\} \bydef \bbE\{R_s[n]\} - \bbE\{\hat{R}_s[n]\} \label{eqn:lossInSumRate} .
\end{equation}
To simplify analysis, we derive an upper bound on the mean loss in sum-rate by deriving first a lower bound on $\bbE\{\hat{R}_s[n]\}$ given by 
\begin{eqnarray}
\bbE\{\hat{R}_s[n]\} &\geq& \sum_{k=1}^K \log_2\left(\frac{\rho_{k}|\bh^T_{k}[n]\hfk[n]|^2}{1 + \sum_{\substack{\l = 1 \\ \l \neq k}}^{K }\alpha_{k,\l}\rho_k|\hgkl^T[n]\hat{\bff}_{\l}[n]|^2} \right) , \label{eqn:MeanRs_LFD_1}\\
& = & \sum_{k=1}^K  \underbrace{\bbE\left\{ \log_2\left(\rho_{k}|\bh^T_{k}[n]\hfk[n]|^2\right) \right\}}_{R_{k,{\rm (des)}}} - \underbrace{\bbE\left\{\log_2\left(1 + \sum_{\substack{\l = 1 \\ \l \neq k}}^{K }\alpha_{k,\l}\rho_k|\bg_{k,\l}^T[n]\hat{\bff}_{\l}[n]|^2 \right) \right\}}_{R_{k,{\rm(int)}}} .\label{eqn:MeanRs_LFD_2}
\end{eqnarray}
We label the first term on the right-hand side of \eqref{eqn:MeanRs_LFD_2} as ${R_{k,{\rm (des)}}}$ and the second term as $R_{k,{\rm(int)}}$. We obtain a lower bound on the mean sum-rate by deriving a lower-bound on ${R_{k,{\rm (des)}}}$ and an upper bound on $R_{k,{\rm(int)}}$. To derive closed-form bounds on ${R_{k,{\rm (des)}}}$ and $R_{k,{\rm(int)}}$ we use RVQ \cite{Jindal2006}, where each codeword is independent and isotropically distributed on a complex unit $\Nt$ dimensional hypersphere. 

\begin{proposition}\label{lem:meanDesChan}
The mean of $\bbE\left\{ \log_2\left(\rho_{k}|\bh^T_{k}[n]\hfk[n]|^2\right) \right\}$ using RVQ is given by
\begin{align}
\begin{split}
\bbE\left\{ \log_2\left(\rho_{k}|\bh^T_{k}[n]\hfk[n]|^2\right) \right\} &\geq  \log_2(\rho_k\eta^2_{k})+ \log_2(e)\sum_{i=0}^{2^\Bk}{2^\Bk \choose i} (-1)^i \sum_{n=1}^{i(\Nt-1)} \frac{1}{n}  \nonumber\\
&  + \bbE\left\{\log_2\left(\|\hk[n-\Dk]\|^2\Big| \hhk[n-\Dk]\hfk[n] \Big|^2\right)\right\}.  \nonumber
\end{split}
\end{align}
\end{proposition}
\begin{proof}
The proof is given in Appendix \ref{app:meanDesChan}.
\end{proof}
Now, $\bbE\{R_s[n]\} \approx \sum_{k=1}^K \bbE\left\{\log_2\left(\rho_k\|\hk[n]\|^2|\thk^T[n]\bff_{k}[n]|^2 \right)\right\}$. Using \eqref{eqn:lossInSumRate} and \eqref{eqn:MeanRs_LFD_2}, Proposition~\ref{lem:meanDesChan} relates the error from quantizing the desired channel using $\Bk$ bits to the mean loss in sum-rate, $\bbE\{\Delta R_s[n]\}$,  
\begin{align}
\begin{split}
\bbE\{\Delta R_s[n]\}  &  \leq \sum_{k=1}^K\log_2(\eta^2_{k})+ \log_2(e)\sum_{i=0}^{2^\Bk}{2^\Bk \choose i} (-1)^i \sum_{n=1}^{i(\Nt-1)} \frac{1}{n}  + R_{k,{\rm(int)}}, \label{eqn:DeltaRs1}
\end{split}
\end{align}
since $\bbE\left\{\log_2\left(\|\hk[n]\|^2|\thk^T[n]\bff_{k}[n]|^2 \right)\right\} = \bbE\left\{\log_2\left(\|\hk[n-\Dk]\|^2\Big| \hhk[n-\Dk]\hfk[n] \Big|^2\right)\right\}$. To further simplify \eqref{eqn:DeltaRs1}, we express $\sum_{i=0}^{2^\Bk}{2^\Bk \choose i} (-1)^i \sum_{n=1}^{i(\Nt-1)} \frac{1}{n}$ as the sum of beta functions, as shown by Lemma~\ref{lem:sumToBetaSum}. 
\begin{lemma}\label{lem:sumToBetaSum}
The contribution of the desired channel delay and quantization towards the loss in sum-rate, $\sum_{i=0}^{2^\Bk}{2^\Bk \choose i} (-1)^i \sum_{n=1}^{i(\Nt-1)} \frac{1}{n}$ is given by
\begin{align}
\begin{split}
\sum_{i=0}^{2^\Bk}{2^\Bk \choose i} (-1)^i \sum_{n=1}^{i(\Nt-1)} \frac{1}{n} = \frac{1}{\Nt-1}\sum_{i=1}^{\Nt-1} \beta\left(2^\Bk, \frac{i}{\Nt-1}\right) .  \nonumber
\end{split}
\end{align}
\end{lemma}
\begin{proof}
The proof is given in Appendix \ref{app:sumToBetaSum}.
\end{proof}

Using Lemma~\ref{lem:sumToBetaSum}, \eqref{eqn:DeltaRs1} is rewritten as
\begin{align}
\begin{split}
\bbE\{\Delta R_s[n]\}  &  \leq \sum_{k=1}^K\log_2(\eta^2_{k})+ \frac{\log_2(e)}{\Nt-1}\sum_{i=1}^{\Nt-1} \beta\left(2^\Bk, \frac{i}{\Nt-1}\right)  + R_{k,{\rm(int)}} . \label{eqn:DeltaRs2}
\end{split}
\end{align}
Intuitively, as $\Bk$ increases, the contribution of the desired channel quantization towards the mean loss in sum-rate reduces.  By definition, we have that $\beta(2^{B_1},x) < \beta(2^{B_2},x)$ for $B_1<B_2$. Hence, Proposition~\ref{lem:meanDesChan} and Lemma \ref{lem:sumToBetaSum} verify that increasing $\Bk$  leading to a smaller value of $R_{k,{\rm (des)}}$.

Using Jensen's inequality, $R_{k,{\rm(int)}}$ can be upper bounded by
\begin{eqnarray}
R_{k,{\rm(int)}} &\leq& \log_2\left(1 + \rho_k\sum_{\substack{\l = 1 \\ \l \neq k}}^{K}\alpha_{k,\l}\bbE\left\{|\bg_{k,\l}^T[n]\hfl[n]|^2\right\}\right). \label{eqn:Tint_1} 
\end{eqnarray}
We evaluate $\bbE\left\{|\gkl^T[n]\hfl[n]|^2\right\}$ to derive an upper bound on $R_{k,{\rm(int)}}$. 

\begin{proposition}\label{lem:meanIntChan}
The mean of $|\gkl^T[n]\hfl[n]|^2$ using RVQ is given by
\begin{align}
\begin{split}
\bbE\left\{|\gkl^T[n]\hfl[n]|^2\right\} &\leq  1-\eta_{k,\l}^2 + \eta^2_{k,\l} 2^\Bkl\beta\left(2^\Bkl,\frac{\Nt}{\Nt-1}\right)\frac{\Nt}{\Nt-1} .  \nonumber
\end{split}
\end{align}
\end{proposition}
\begin{proof}
The proof is given in Appendix \ref{app:meanIntChan}.
\end{proof}
Proposition \ref{lem:meanIntChan} relates the mean loss in sum-rate to the number of bits assigned to quantize the interfering channels, $\Bkl$. Intuitively, increasing $\Nt$ by keeping $\Bkl$ fixed will lead to an increase in the quantization error, which will increase the mean loss in sum-rate. This is verified using Proposition \ref{lem:meanIntChan}, where it is seen that for large $\Nt$, $\Nt/(\Nt-1) \rightarrow 1$, causing $\beta\left(2^\Bkl,\frac{\Nt}{\Nt-1}\right)\frac{\Nt}{\Nt-1}$ to become larger. Using \eqref{eqn:DeltaRs2} and Proposition \ref{lem:meanIntChan}, $\bbE\{\Delta R_s[n]\} $ is upper bounded as
\begin{align}
\begin{split}
\bbE\{\Delta R_s[n]\}  &  \leq \sum_{k=1}^K\log_2(\eta^2_{k})+ \frac{\log_2(e)}{\Nt-1}\sum_{n=1}^{\Nt-1} \beta\left(2^\Bk, \frac{n}{\Nt-1}\right) \\
 & + \log_2\left(1+\rho_k\sum_{\substack{\l = 1 \\ \l \neq k}}^{K}\alpha_{k,\l} \left( 1-\eta_{k,\l}^2 +\eta^2_{k,\l} 2^\Bkl\beta\left(2^\Bkl,\frac{\Nt}{\Nt-1}\right)\frac{\Nt}{\Nt-1}\right)\right) \label{eqn:lossInSumRate-preFinal2} .
\end{split}
\end{align}
It is clear from \eqref{eqn:lossInSumRate-preFinal2} that for a given $\{\rho_k,\eta_k\}_{k=1}^K$ and $\{\eta_{k,\l},\alpha_{k,\l}\}_{\substack{k=1\\ k \neq \l}}^K$, the (upper bound on) mean loss in sum-rate will depend on $\Bk$ and $\Bkl$, for all $k$ and $\l$. 

For $\Nt > K$ (where $\Bk \neq 0$), one solution to choosing $\Bk$ and $\Bkl$ is to partition the feedback bits as $\Bk = \Bkl = \Bt/K$. Note that while this solution is simple, it does not account for the different weights associated with each interferer. For example, for a given user $k$, if $\alpha_{k,\l} \approx 0$ for all $\l$, it implies that the interfering signals to the user are too weak to contribute to the mean loss in sum-rate and hence, can be ignored completely. Equal-bit allocation, however, will assign $\Bkl = \Bt/K$ to each interferer and waste $ (K-1)\Bt/K$ bits, while assigning only $\Bt/K$ to quantize the desired channel. Thus, an \emph{adaptive} feedback bit partitioning strategy will efficiently allocate bits to the different channels depending on the respective delays and signal strengths.  

\section{Feedback-Bit Partitioning} \label{sec:FB_BitPart}
In this section we minimize the upper bound on the mean loss in sum-rate in \eqref{eqn:lossInSumRate-preFinal2}, with respect to the number of feedback bits assigned to the desired and interfering channels at each user. We simplify \eqref{eqn:lossInSumRate-preFinal2} using approximations for the beta functions to ensure analytical tractability and then derive closed-form expressions for the bit-partitioning.

A beta function $\beta(a,b)$ ($=\int_0^1 t^{a-1}(1-t)^{b-1}dt$) can be approximated by $\beta(a,b) \approx \Gamma(b)a^{-b},$ when $a$ is large and $b$ is fixed. Since we are optimizing over $\Bk$ and $\Bkl$ in \eqref{eqn:lossInSumRate-preFinal2}, $b$ is always fixed while $a = 2^\Bk$ (or $2^\Bkl$) is at least equal to 1. Hence, we approximate the right hand side of the expression in \eqref{eqn:lossInSumRate-preFinal2} as
\begin{align}
\begin{split}
\Delta_R  &  \approx \sum_{k=1}^K\log_2(\eta^2_{k})+ \frac{\log_2(e)}{\Nt-1}\sum_{n=1}^{\Nt-1} \Gamma\left( \frac{n}{\Nt-1}\right) \left(2^\Bk\right)^{-\frac{n}{\Nt-1}} \\
 & + \log_2\left(1+\rho_k\sum_{\substack{\l = 1 \\ \l \neq k}}^{K}\alpha_{k,\l} \left( 1-\eta_{k,\l}^2 +\eta^2_{k,\l} 2^\Bkl\Gamma\left( \frac{\Nt}{\Nt-1}\right) \left(2^\Bkl\right)^{-\frac{\Nt}{\Nt-1}}\frac{\Nt}{\Nt-1}\right)\right) \label{eqn:lossInSumRate-approx1} .
\end{split}
\end{align}
For a large $\Nt$, $\left(2^\Bk\right)^{-\frac{n}{\Nt-1}}$ will be relatively much smaller than $\left(2^\Bk\right)^{-\frac{1}{\Nt-1}}$ and hence, can be ignored. By neglecting the higher order terms ($n=2,\ldots,(\Nt-1)$) in the summation
\begin{align}
\begin{split}
\frac{\log_2(e)}{\Nt-1}\sum_{n=1}^{\Nt-1} \Gamma\left( \frac{n}{\Nt-1}\right) \left(2^\Bk\right)^{-\frac{n}{\Nt-1}} \approx \frac{\log_2(e)}{\Nt-1} \Gamma\left( \frac{1}{\Nt-1}\right) \left(2^\Bk\right)^{-\frac{1}{\Nt-1}}.  \label{eqn:lossInSumRate-approx2}
\end{split}
\end{align}
The mean loss in sum-rate can then be further approximated as
\begin{align}
\begin{split}
\Delta_R  &  \approx \sum_{k=1}^K\log_2(\eta^2_{k})+\log_2(e)\Gamma\left( \frac{\Nt}{\Nt-1}\right)2^{-\frac{\Bk}{\Nt-1}} \\
 & + \log_2\left(1+ \rho_k\sum_{\substack{\l = 1 \\ \l \neq k}}^{K}\alpha_{k,\l} \left( 1-\eta_{k,\l}^2\right) +\Gamma\left( \frac{2\Nt-1}{\Nt-1}\right) \Nt\rho_k\sum_{\substack{\l = 1 \\ \l \neq k}}^{K}\alpha_{k,\l}\eta^2_{k,\l}  2^{-\frac{\Bkl}{\Nt-1}}\right), \label{eqn:lossInSumRate-approx4}
\end{split}
\end{align}
using the identity $N\Gamma(N) = \Gamma(N+1)$.
We denote 
\begin{align}
\begin{split}
\Delta_{R,k}(\Bk,\{\Bkl\}_\l)  \bydef &\log_2(e)\Gamma\left( \frac{\Nt}{\Nt-1}\right)2^{-\frac{\Bk}{\Nt-1}} \\
& + \log_2\left(1+ \Nt\rho_k\sum_{\substack{\l = 1 \\ \l \neq k}}^{K}\alpha_{k,\l} \left( 1-\eta_{k,\l}^2\right) +\Gamma\left( \frac{2\Nt-1}{\Nt-1}\right) \Nt\rho_k\sum_{\substack{\l = 1 \\ \l \neq k}}^{K}\alpha_{k,\l}\eta^2_{k,\l}  2^{-\frac{\Bkl}{\Nt-1}}\right)\label{eqn:lossInSumRate-approx5} .
\end{split}\end{align}

The mean loss in sum-rate can finally be written as
\begin{equation}
\Delta_R   \approx \sum_{k=1}^K\log_2\left(\eta^2_{k}\right)  + \sum_{k=1}^K\Delta_{R,k}(\Bk,\{\Bkl\}_{\substack{\l=1\\\l\neq k}}^K) .\label{eqn:lossInSumRate-approx6}
\end{equation}
For a fixed $\eta_k$, minimizing \eqref{eqn:lossInSumRate-approx6} reduces to minimizing $\sum_{k=1}^K\Delta_{R,k}(\Bk,\{\Bkl\}_{\substack{\l=1\\\l\neq k}}^K)$, which is equivalent to minimizing each $\Delta_{R,k}(\Bk,\{\Bkl\}_{\substack{\l=1\\\l\neq k}}^K)$ individually due to the independence of each term in the summation. Note that this eliminates the need for joint optimization. 

The optimization problem to minimize the $k^{\rm th}$ user's contribution to the mean loss in sum-rate, $\Delta_{R,k}(\Bk,\{\Bkl\}_{\substack{\l=1\\\l\neq k}}^K)$ (and hence, the expression in \eqref{eqn:lossInSumRate-approx5}) is given by
\begin{align}
\begin{split}
&\min_{\Bk,\{\Bkl\}_{\substack{\l=1\\\l\neq k}}^K}\Delta_{R,k}(\Bk,\{\Bkl\}_{\substack{\l=1\\\l\neq k}}^K) \\ 
\rm{s.t.}&~\Bk + \sum_{\substack{\l=1\\ \l\neq k}}^K \Bkl = \Bt, ~{\rm and}~ \Bk, \Bkl \geq 0 \label{eqn:OptProb}.
\end{split}
\end{align}
We denote the total number of bits allocated to quantize interfering channels by $\Bi = \sum_{\substack{\l=1\\ \l\neq k}}^K \Bkl$, where $\Bk+\Bi = \Bt$. Given $\Bi$, we first derive $\Bkl$ to minimize the contribution of the interfering channels towards the mean loss in sum-rate, i.e., we minimize
\begin{equation}
\log_2\left(1+\rho_k\sum_{\substack{\l = 1 \\ \l \neq k}}^{K}\alpha_{k,\l} \left( 1-\eta_{k,\l}^2\right) +\Gamma\left( \frac{2\Nt-1}{\Nt-1}\right)\rho_k \Nt\sum_{\substack{\l = 1 \\ \l \neq k}}^{K}\alpha_{k,\l}\eta^2_{k,\l}  2^{-\frac{\Bkl}{\Nt-1}}\right) \label{eqn:OptProb_IntMin-a}
\end{equation}
such that $\sum_{\substack{\l=1\\ \l\neq k}}^K \Bkl = \Bi$. Since the log function is monotonic in nature and $\{\Bkl\}_{\substack{\l=1\\\l\neq k}}^K$ are the only variables in \eqref{eqn:OptProb_IntMin-a}, the optimization problem in \eqref{eqn:OptProb_IntMin-a} is reduced to
\begin{align}
\begin{split}
&\min_{\{\Bkl\}_{\substack{\l=1\\\l\neq k}}^K}~~\sum_{\substack{\l = 1 \\ \l \neq k}}^{K}\alpha_{k,\l}\eta^2_{k,\l}  2^{-\frac{\Bkl}{\Nt-1}}\\ 
\rm{s.t.}&~\sum_{\substack{\l=1\\ \l\neq k}}^K \Bkl = \Bi, ~{\rm and}~ \Bkl \geq 0 \label{eqn:OptProb-IntMinFinal}.
\end{split}
\end{align}
The solution to \eqref{eqn:OptProb-IntMinFinal} and the contribution of each of the terms inside the summation in the objective function of \eqref{eqn:OptProb-IntMinFinal} is given by Theorem~\ref{thm:OptSolnIntChan} using the arithmetic-geometric mean inequality. We consider first unconstrained optimization, where the bits $\Bk$ and $\Bkl$ do not have to be integers. We later use convexity arguments to determine non-negative integer solutions to $\Bk$ and $\Bkl$.

\begin{theorem}\label{thm:OptSolnIntChan}
The optimum number of bits assigned to the $\l^{\rm th}$ interferer, $\Bkl^*$, that minimizes \eqref{eqn:OptProb-IntMinFinal} is given by
$$
\Bkl^* = \frac{\Bi}{|\cK|} + (\Nt-1)\log_2\left(\frac{\alpha_{k,\l}\eta^2_{k,\l} }{\prod_{\l \in\cK}(\alpha_{k,\l}\eta^2_{k,\l})^{\frac{1}{|\cK|}}}\right),
$$
for $\l\in\cK$, and $\Bkl = 0$ for $\l\notin\cK$, where $\cK$ is the largest set of interferers that satisfies
$$
\log_2\left(\frac{\prod_{\l \in\cK}(\alpha_{k,\l}\eta^2_{k,\l})^{\frac{1}{|\cK|}}}{\alpha_{k,\l}\eta^2_{k,\l}} \right) < \frac{\Bt}{|\cK|(\Nt-1)} .
$$
\end{theorem}
\begin{proof}
The proof is given in Appendix \ref{app:OptSolnIntChan}.
\end{proof}

In Theorem \ref{thm:OptSolnIntChan}, $\cK$ denotes the set of \emph{effective} set of interferers. Only the effective interfering channels are quantized using the proposed bit-partitioning algorithm. The idea is that the limited number of feedback bits are best utilized by allocating them to strong interferers, which can affect the sum-rate drastically or/and to the interferers with small delays, since nulling out outdated CSI will not improve sum-rates. Feedback bits are assigned to the channels in proportion to their signal strengths and delays. If all the interferers have the same delay and strength, then it is seen from Theorem \ref{thm:OptSolnIntChan} that all the channels are assigned the same number of feedback bits. 

Using the bit allocation strategy for interfering channels proposed in Theorem~\ref{thm:OptSolnIntChan}, the minimum value of the objective function in \eqref{eqn:OptProb-IntMinFinal} is given by
$$
\min_{\{\Bkl\}_{\substack{\l=1\\ \l\neq k}}^K}~~\sum_{\substack{\l = 1 \\ \l \neq k}}^{K}\alpha_{k,\l}\eta^2_{k,\l}  2^{-\frac{\Bkl}{\Nt-1}} = |\cK|2^{-\frac{\Bi}{|\cK|(\Nt-1)}}\prod_{\l\in\cK}(\alpha_{k,\l}\eta^2_{k,\l})^{\frac{1}{|\cK|}}.
$$
The objective function of \eqref{eqn:OptProb}, $\Delta_{R,k}(\Bk,\{\Bkl\}_{\substack{\l=1\\\l\neq k}}^K)$, can then be rewritten as
\begin{align}
\begin{split}
\Delta_{R,k}&(\Bk,\{\Bkl\}_{\substack{\l=1\\\l\neq k}}^K)  = \log_2(e)\Gamma\left( \frac{\Nt}{\Nt-1}\right)2^{-\frac{\Bk}{\Nt-1}} \\
& +\log_2\left(1 + \Nt\rho_k\sum_{\substack{\l = 1 \\ \l \neq k}}^{K}\alpha_{k,\l} \left( 1-\eta_{k,\l}^2\right)+|\cK|\rho_k \Nt\Gamma\left(\frac{2\Nt-1}{\Nt-1}\right)2^{-\frac{\Bt-\Bk}{|\cK|(\Nt-1)}}\prod_{\l\in\cK}(\alpha_{k,\l}\eta^2_{k,\l})^{\frac{1}{|\cK|}}\right) \\
& = \Delta_{R,k}(\Bk). \label{eqn:OptProb2}
\end{split}
\end{align}
Hence, the objective function of \eqref{eqn:OptProb} is now reduced to an expression in a single variable. To simplify the function in \eqref{eqn:OptProb2}, we consider different solutions for the high and low SNR regimes. 

\noindent \textbf{\emph{Case 1: Low SNR Regime}}\\
Using the approximation that $\ln(1+x) \approx x$ for $x \approx 0$, $\Delta_{R,k}(\Bk)$ can be written as
\begin{align}
\begin{split}
\Delta_{R,k}(\Bk)  \approx &\log_2(e)\Gamma\left( \frac{\Nt}{\Nt-1}\right)2^{-\frac{\Bk}{\Nt-1}} +\log_2(e)\rho_k \Nt\sum_{\substack{\l = 1 \\ \l \neq k}}^{K}\alpha_{k,\l} \left( 1-\eta_{k,\l}^2\right) \\
& +\log_2(e)|\cK| \Nt\rho_k\Gamma\left(\frac{2\Nt-1}{\Nt-1}\right)2^{-\frac{\Bt-\Bk}{|\cK|(\Nt-1)}}\prod_{\l\in\cK}(\alpha_{k,\l}\eta^2_{k,\l})^{\frac{1}{|\cK|}}. \label{eqn:OptProb3-LowSNR}
\end{split}
\end{align}
We use the arithmetic mean-geometric mean inequality to minimize the expression in \eqref{eqn:OptProb3-LowSNR}, as given in Theorem~\ref{thm:SolnLowSNR}. 

\begin{theorem}\label{thm:SolnLowSNR}
Given the total number of bits allocated to quantize \emph{all} the channels, $\Bt$, the optimum number of bits assigned to the desired channel at the $k^{\rm th}$ user, $\Bk$, to minimize \eqref{eqn:OptProb3-LowSNR} is given by
$$
\Bk = \frac{\Bt}{|\cK|+1} - \frac{(\Nt-1)|\cK|}{|\cK|+1} \log_2\left(\rho_k\frac{\Nt}{\Nt-1}\prod_{\l \in \cK}(\alpha_{k,\l}\eta^2_{k,\l})^{\frac{1}{|\cK|}}\right),                
$$
for $\Nt > K$ and $\Bk = 0$ for $\Nt = K$. The optimum number of bits assigned to all the interfering channels at the $k^{\rm th}$ user is computed as $\Bi=\Bt-\Bk$. 
\end{theorem}
\begin{proof}
The proof is given in Appendix \ref{app:SolnLowSNR}.
\end{proof}

From Theorem~\ref{thm:SolnLowSNR}, it is seen that as $\rho_k\prod_{\l \in \cK}\alpha_{k,\l}\eta^2_{k,\l}$ increases (or as interferers become stronger/ have smaller delays) $\Bk$ is reduced. In the presence of strong interference, reducing the interfering signal power will increase the sum-rate more than further improving the already strong desired signal strength. In contrast, when the interfering signals are weak, i.e. $\rho_k\prod_{\l \in \cK}\alpha_{k,\l}\eta^2_{k,\l}$ is small, most bits will be assigned to quantize the desired channel more finely in an attempt to improve the desired signal power. Hence, the theorem makes intuitive sense. 

\noindent \textbf{\emph{Case 2: High SNR Regime}}\\
In the high SNR regime, we use the approximation that $\ln(1+x) \approx \ln x$ for $x >> 1$ to approximate $\Delta_{R,k}(\Bk)$ in \eqref{eqn:OptProb2} by 
\begin{align}
\begin{split}
\Delta_{R,k}(\Bk)  \approx &\log_2(e)\Gamma\left( \frac{\Nt}{\Nt-1}\right)2^{-\frac{\Bk}{\Nt-1}} +\log_2\left(\rho_k\sum_{\substack{\l = 1 \\ \l \neq k}}^{K}\alpha_{k,\l} \left( 1-\eta_{k,\l}^2\right)\right. \\
& \left.+\rho_k\Gamma\left( \frac{2\Nt-1}{\Nt-1}\right)2^{-\frac{\Bt-\Bk}{(\Nt-1)(K-1)}}\prod_{\l = 1}^{K}(\alpha_{k,\l}\eta^2_{k,\l})^{\frac{1}{K-1}}\right) . \label{eqn:OptProb3-HighSNR}
\end{split}
\end{align}
The optimal number of bits to minimize \eqref{eqn:OptProb3-HighSNR} are given in Theorem~\ref{thm:SolnHighSNR}.

\begin{theorem}\label{thm:SolnHighSNR}
Given the total number of bits allocated to quantize \emph{all} the channels, $\Bt$, the optimum number of bits assigned to the desired channel at the $k^{\rm th}$ user, $\Bk$, is given by
$$
\Bk = (\Nt -1) \log_2\left( (|\cK|-1)\Gamma\left(\frac{\Nt}{\Nt-1}\right)\right). 
$$
The optimum number of bits assigned to all the interfering channels at the $k^{\rm th}$ user, $\Bi = \Bt  - \Bk$.
\end{theorem}
\begin{proof}
The proof is given in Appendix \ref{app:SolnHighSNR}.
\end{proof}

It is seen from Theorem~\ref{thm:SolnHighSNR} that $\Bk$ and $\Bi$ are independent of the SNR, in the high SNR regime. Further, $\Bk$ is also independent of the received signal strengths and delays of the interferers. This makes intuitive sense if we take into account that the sum-rate saturates as SNR $\rightarrow\infty$. Hence, minimizing the mean loss in sum-rate at high SNR will be independent of the SNR of the desired (and interfering) channels.   

Note that the feedback bit assignments in Theorems~\ref{thm:OptSolnIntChan}, \ref{thm:SolnLowSNR}, and \ref{thm:SolnHighSNR} are not necessarily integers. To ensure that $\Bk$ and $\Bkl$ are positive integers, we only have to check for the ceiling and floor of $\Bk$ and $\Bkl$ due to the convexity of the objective functions in \eqref{eqn:OptProb3-LowSNR} and \eqref{eqn:OptProb-IntMinFinal}, respectively. In \eqref{eqn:OptProb3-LowSNR}, the objective function is the sum of two non-negatively weighted convex functions, $2^{-\frac{\Bk}{\Nt-1}}$ and $2^{\frac{\Bk}{|\cK|(\Nt-1)}}$ and hence, is convex. By the same argument, \eqref{eqn:OptProb-IntMinFinal} is also convex. This is similar to the approach in \cite{Bhagavat2009}. 

The approximate mean loss in sum-rate in \eqref{eqn:OptProb3-HighSNR} can be rewritten as 
\begin{align}
\begin{split}
\log_2(e)\Gamma\left( \frac{\Nt}{\Nt-1}\right)2^{-\frac{\Bk}{\Nt-1}} +&\log_2\left(1 + \frac{\Gamma\left( \frac{2\Nt-1}{\Nt-1}\right)2^{\frac{-\Bt}{(\Nt-1)(K-1)}}\prod_{\l = 1}^{K}(\alpha_{k,\l}\eta^2_{k,\l})^{\frac{1}{K-1}}}{\sum_{\substack{\l = 1 \\ \l \neq k}}^{K}\alpha_{k,\l} \left( 1-\eta_{k,\l}^2\right)}2^{\frac{\Bk}{(\Nt-1)(K-1)}}\right) \\
& - \log_2\left(\rho_k\sum_{\substack{\l = 1 \\ \l \neq k}}^{K}\alpha_{k,\l} \left( 1-\eta_{k,\l}^2\right)\right). \label{eqn:OptProb3-HighSNR-2}
\end{split}
\end{align}
The last term can be ignored in the optimization problem since it is independent of $\Bk$. Denoting the ratio before $2^{\frac{\Bk}{(\Nt-1)(K-1)}}$ inside the logarithm by $C_i$, the objective function in \eqref{eqn:OptProb3-HighSNR-2} can be simplified as
\begin{align}
\begin{split}
\log_2(e)\Gamma\left( \frac{\Nt}{\Nt-1}\right)2^{-\frac{\Bk}{\Nt-1}} +\log_2\left(1 +C_i2^{\frac{\Bk}{(\Nt-1)(K-1)}}\right) . \label{eqn:OptProb3-HighSNR-3}
\end{split}
\end{align}
Note that $C_i$ will typically be large since the denominator of $\sum_{\substack{\l = 1 \\ \l \neq k}}^{K}\alpha_{k,\l} \left( 1-\eta_{k,\l}^2\right)$ is very small.

\section{Simulation Results} \label{sec:Results}
In this section, we present simulation results to demonstrate the gain in the mean per-cell data rate obtained using the proposed adaptive feedback-bit partitioning algorithm, over equal-bit allocation. We consider a seven-cell system with a single active user per cell. Each base station has eight antennas ($\Nt = 8$) and each user has a single antenna. For simulation purposes we focus on a target user located in the cell at the center of the seven-cell grid, which receives the desired signal from its own base station and interference from the six neighboring cells as shown in Fig.~\ref{fig:SimsModel} and from base stations out of the seven-cell system. The mobile terminal is assumed to travel in a straight line from the cell center to the cell edge with a velocity $v$. The distance between the user and the desired base station is given by $d\leq R$. The system setup is based on the urban microcell propagation scenario in the 3GPP spatial channel model \cite{3GPP2003}. The radius of each cell, R, is assumed to be $400~m$. The path loss between the base stations and the mobile user is modeled using the COST 231 Walfish-Ikegami NLOS model \cite{3GPP2003}, adopted for urban mircocells. Using a carrier frequency of $1.9~GHz$, base station antenna height of $12.5~m$, mobile terminal height of $1.5~m$, building height $12~m$, building to building distance $50~m$ and street width $25~m$, the path-loss in $dB$, $PL[dB]$, is given by \cite{3GPP2003}
MS antenna height 1.5m 
\begin{equation}
PL[dB] = 34.53 + 38\log_{10}(d) \label{eqn:SimsPLModel}
\end{equation}
The transmit power, $E_s = 3~dBW$ for all the base stations and the noise power is given by $-144~dBW$. We also model the delay associated with the feedback of CSI to the desired base station by one symbol time and that with the exchange of CSI over the backhaul link connecting the desired and interfering base stations by two symbol times. The parameters used for simulations are tabulated in Table~\ref{tab:simParams}.

In Fig.~\ref{fig:dataRateVsDist}, we compare the performance of equi-bit partitioning and the proposed adaptive strategy as a function of the distance from the desired base station. In the figure, we normalize the data rates from the two limited feedback techniques by the full CSI data rate for $\Bt = 7$ and $\Bt = 35$ with $v=10~mph$. It is seen that the feedback-bit assignment presented in this paper outperforms equal-bit allocation irrespective of where the user is located. At the cell-edge, adaptive bit assignment achieves about $45\%$ more of the full CSI data rate as compared to uniform allocation.

The bit partitioning corresponding to $\Bt = 35$ are plotted in Figs.~\ref{fig:bitPartBtot35HighSNR} and \ref{fig:bitPartBtot35LowSNR} as a function of distance for the simulation setup described in Fig.~\ref{fig:SimsModel} and Table~\ref{tab:simParams}. We consider both the high SNR and low SNR partitions by setting $E_s = 3~dB$ and $-3~dB$, respectively. The channel from base station $j$ to the user `0' (as shown in Fig.~\ref{fig:SimsModel}) is quantized using $B_{0,j}$ bits. In Fig.~\ref{fig:bitPartBtot35HighSNR}, it is seen that more bits are assigned to quantize the desired channel towards the cell edge as compared to the low SNR case in Fig.~\ref{fig:bitPartBtot35LowSNR}. This makes intuitive sense because at high SNR, the quantization error associated with the desired channel is larger than that at low SNR. Hence, more bits need to be assigned to quantize the desired channel at high SNR. Another interesting point from Figs.~\ref{fig:bitPartBtot35HighSNR} and \ref{fig:bitPartBtot35LowSNR} is the decrease in the number of bits allocated to the interfering channels from base stations 1, 4, 5 and 6 as the user moves to the cell edge towards cells 2 and 3. This clearly illustrates the adaptive nature of the proposed algorithm, which allocates bits as a function of the received signal strength and the delays.

%

We plot the cell-edge data rates using both, equi-bit allocation and the proposed adaptive partitioning, in Fig.~\ref{fig:dataRateVsBtot} for the system setup described in this section. Note that at the cell edge, the strongest interference is from $BTS~2$ and $BTS~3$, as shown in Fig.~\ref{fig:SimsModel}. The corresponding bit assignments are also given in the figure, for each $\Bt$ in the format of $(B_0,B_{0,1},B_{0,2},B_{0,3},B_{0,4},B_{0,5},B_{0,6})$. Finally, the effect of delay in the backhaul link is plotted in Fig.~\ref{fig:delay} for a user at cell-edge with $\Bt = 7$ and $\Bt = 35$. We fix the feedback delay, $\Dk$ to one symbol time and vary the backhaul delay between $[0,5]$ symbol times, i.e. $\Dkl\in[1,6]$. Similar to Fig.~\ref{fig:dataRateVsBtot}, it is seen that while the proposed algorithm outperforms equal bit allocation for all the delays and $\Bt$ values. 

It is seen from both Fig.~\ref{fig:dataRateVsBtot} and Fig.~\ref{fig:delay} that while the limited feedback technique in this paper outperforms EBA for all $\Bt$, the improvement in data rate is larger for higher $\Bt$. The new strategy yields about $40~\%$ higher sum-rates than the equal-bit approach at $\Bt = 35$ in both Fig.~\ref{fig:dataRateVsBtot} and Fig.~\ref{fig:delay}. This can be explained as follows. At $\Bt = 7$, the desired channel is given all the 7 bits, implying 0 bits for each of the strong interfering channels. For $\Bt = 35$, the strong interfering channels are assigned $14$ bits for each of two strong interfering channels. EBA, in contrast, sees an increase in the feedback bits for the strong interfering channels from $1$ to $5$ bits per channel. Quantizing the strong interferers more finely at the cost of allocating zero bits to the weak channels leads to the significant improvement in data rate using the proposed algorithm. This leads to a better quantization of the strong interfering channels, leading to a significantly better performance. We can also see from Fig.~\ref{fig:dataRateVsBtot} that all the bits are assigned to the desired channel at $\Bt = 7$. The reason is that when the received signal strength for the desired signal is low, the proposed algorithm will first concentrate on improving the desired signal strength through finer quantization, rather than reducing the strength of the interfering signals, which can only be as strong as the desired signal.

\section{Conclusion} \label{sec:Conclusion}
We considered a multicell cooperative MISO system using ICIN. We assumed a single active user per cell, which estimates and feeds back the desired and interfering CSI to its own base station. Backhaul links are used to ensure that each base station has knowledge of the interference that it is causing to neighboring cells. The feedback and backhaul links are assumed to have delays associated with them. In this paper, we quantified the mean loss in sum-rate due to delayed limited feedback of the desired and interfering CSI, using ICIN. We derived a closed-form expression for allocating feedback bits to quantize the multiple channels. We showed, using simulations, that the proposed algorithm yields higher mean sum-rates as compared to the equal-bit partitioning approach. 


\appendices

\section{Proof of Proposition~\ref{lem:meanDesChan}}\label{app:meanDesChan}
To evaluate the sum-rate in \eqref{eqn:Rs_LFD}, the relation between $\hk[n]$ and $\hfk[n]$ needs to be known. Since the $k^\th$ base station uses the delayed quantized desired channel, $\hhk[n-\Dk]$, to design $\hfk[n]$, we first determine the dependence of $\hk[n]$ and $\hhk[n-\Dk]$. While the Gauss-Markov model in \eqref{eqn:GM_hk} is used to relate $\hk[n]$ and $\hk[n-\Dk]$, the relationship between $\thk[n-\Dk]$ and $\hhk[n-\Dk]$ is given by\footnote{In the remainder of this proof, we refer to $\cos(\theta_{\thk[n-\Dk], \hhk[n-\Dk]})$ as $\cos(\theta_{\thk, \hhk})$, and $\sin(\theta_{\thk[n-\Dk], \hhk[n-\Dk]})$ as $\sin(\theta_{\thk, \hhk})$ for the sake of notational brevity.}\cite{Jindal2006}
\begin{equation}
\thk[n-\Dk] = \cos(\theta_{\thk, \hhk}) \hhk[n-\Dk] + \sin(\theta_{\thk, \hhk}) \sk[n], \label{eqn:RelnTildeHat-h}
\end{equation}
where $\sk[n]$ is an isotropically distributed vector in the null-space of $\hhk[n-\Dk]$, and is independent of $\cos(\theta_{\thk, \hhk})$ (or $\sin(\theta_{\thk, \hhk})$). Substituting \eqref{eqn:RelnTildeHat-h} in \eqref{eqn:GM_hk} yields
\begin{equation} \small
\hk[n] = \eta_{k}\|\hk[n-\Dk]\|\left(\cos(\theta_{\thk, \hhk}) \hhk[n-\Dk] + \sin(\theta_{ \thk, \hhk})\sk[n]  \right) + \sqrt{1-\eta_{k}^2} \wk[n]. \label{eqn:htimeCorrApp}
\end{equation} 
Computing $|\hk^T[n]\hfk[n]|^2$, we have 
\begin{eqnarray} 
|\hk^T[n]\hfk[n]|^2 & = & \Big|\eta_{k}\|\hk[n-\Dk]\|\left(\cos(\theta_{\thk, \hhk}) \hhk[n-\Dk]\hfk[n]   
+\sin(\theta_{ \thk, \hhk})\sk[n] \hfk[n] \right) \nonumber \\  
& &+ \sqrt{1-\eta_{k}^2}\wk[n]\hfk[n] \Big|^2. \nonumber
\end{eqnarray} 
Since $\sin(\theta_{\thk, \hhk})$ and $\sqrt{1-\eta_{k}^2}~(<< 1)$ are generally very small, we can approximate $|\hk^T[n]\hfk[n]|^2$ by 
\begin{equation} 
|\hk^T[n]\hfk[n]|^2  \approx \eta^2_{k}\|\hk[n-\Dk]\|^2\cos^2(\theta_{\thk, \hhk})\Big| \hhk[n-\Dk]\hfk[n] \Big|^2 .\label{eqn:evalHF_2} 
\end{equation} 
Hence, $\bbE\{\log_2(\rho_k|\hk^T[n]\hfk[n]|^2)\}$ can be written as
\begin{align}
\begin{split}
\bbE&\{\log_2(|\hk^T[n]\hfk[n]|^2)\}  \geq \bbE\left\{\log_2\left(\rho_k\eta^2_{k}\|\hk[n-\Dk]\|^2\cos^2(\theta_{\thk, \hhk})\Big| \hhk[n-\Dk]\hfk[n] \Big|^2\right)\right\} \\
&= \log_2(\eta^2_{k})+ \log_2(e)\sum_{i=0}^{2^\Bk}{2^\Bk \choose i} (-1)^i \sum_{n=1}^{i(\Nt-1)} \frac{1}{n}  + \bbE\left\{\log_2\left(\|\hk[n-\Dk]\|^2\Big| \hhk[n-\Dk]\hfk[n] \Big|^2\right)\right\},\label{eqn:evalHF_5} 
\end{split}
\end{align}
where \eqref{eqn:evalHF_5}  is obtained using a result from \cite{Bhagavat2009}, that $
\bbE\left\{\ln\left(\cos^2(\theta_{\thk, \hhk})\right)\right\} =\sum_{i=0}^{N}{N \choose i} (-1)^i \sum_{n=1}^{i(\Nt-1)} \frac{1}{n} $.

\section{Proof of Lemma~\ref{lem:sumToBetaSum}}\label{app:sumToBetaSum}
Denoting $\Nt-1 = M$, $\sum_{i=1}^{2^\Bk}{2^\Bk \choose i} (-1)^i \sum_{n=1}^{iM} \frac{1}{n}$ can be written as
\begin{equation}
\sum_{i=1}^{2^\Bk}{2^\Bk \choose i} (-1)^i \sum_{n=1}^{iM} \frac{1}{n} = \sum_{i=1}^{2^\Bk}{2^\Bk \choose i} (-1)^i \int_0^1 \frac{1-x^{iM}}{1-x} dx. \label{eqn:sumToBetaSum_1}
\end{equation}
Since $\frac{1-x^{iM}}{1-x} \geq 0$, $i=1,\ldots,2^\Bk$ and $x\in[0,1]$, we exchange the integrand and summations to yield
\begin{eqnarray}
\sum_{i=1}^{2^\Bk}{2^\Bk \choose i} (-1)^i \sum_{n=1}^{iM} \frac{1}{n} &=& \int_0^1\frac{1}{1-x} \sum_{i=1}^{2^\Bk}{2^\Bk \choose i} (-1)^i  (1-x^{iM}) dx \label{eqn:sumToBetaSum_2} \\
&=& \int_0^1\frac{1}{1-x} \sum_{i=1}^{2^\Bk}\left[{2^\Bk \choose i} (-1)^i  - {2^\Bk \choose i} (-x^M)^{i}\right]dx \label{eqn:sumToBetaSum_3} \\
&=& -\int_0^1\frac{(1-x^M)^{2^\Bk}}{1-x}  dx \label{eqn:sumToBetaSum_4} ,
\end{eqnarray}
where \eqref{eqn:sumToBetaSum_4} was obtained using $\sum_{i=1}^{2^\Bk}{2^\Bk \choose i} (-1)^i = (1+(-1))^{2^\Bk} = 0$ and $\sum_{i=1}^{2^\Bk}{2^\Bk \choose i}(-x^M)^{i} = (1-x^M)^{2^\Bk}$. The relation $1-x^M = (1-x)\sum_{k=0}^{M-1}x^k$ is used to further simplify \eqref{eqn:sumToBetaSum_4} as
\begin{eqnarray}
\sum_{i=1}^{2^\Bk}{2^\Bk \choose i} (-1)^i \sum_{n=1}^{iM} \frac{1}{n} 
&=& -\int_0^1(1-x^M)^{2^\Bk - 1}\sum_{k=0}^{M-1}x^k  dx \label{eqn:sumToBetaSum_6} .
\end{eqnarray}
Since $x^k \geq 0$ for all $k=1,\ldots,M-1$ and $x\in[0,1]$, we again interchange the summation and integrals
\begin{eqnarray}
\sum_{i=1}^{2^\Bk}{2^\Bk \choose i} (-1)^i \sum_{n=1}^{iM} \frac{1}{n} &=& -\sum_{k=0}^{M-1}\int_0^1(1-x^M)^{2^\Bk - 1}x^k  dx \label{eqn:sumToBetaSum_7}\\ 
&=& -\sum_{k=0}^{M-1}\int_0^1u^{\frac{k+1}{M}-1}(1-u)^{2^\Bk - 1}  du \label{eqn:sumToBetaSum_8} ,
\end{eqnarray}
where $u = x^M$. Note that the expression in \eqref{eqn:sumToBetaSum_8} is that of a beta function with parameters given by $\frac{k+1}{M}$ and $2^\Bk$. Hence, \eqref{eqn:sumToBetaSum_8} is given by the sum of beta functions
\begin{eqnarray}
\sum_{i=1}^{2^\Bk}{2^\Bk \choose i} (-1)^i \sum_{n=1}^{iM} \frac{1}{n} &=& -\sum_{k=1}^{\Nt-1}\beta\left(2^\Bk,\frac{k}{\Nt-1} \right) \label{eqn:sumToBetaSum_8} .
\end{eqnarray}

\section{Proof of Proposition~\ref{lem:meanIntChan}}\label{app:meanIntChan}
Similar to \eqref{eqn:htimeCorrApp}, we can express $\gkl[n]$ as a function of $\hgkl[n-\Dkl]$ as
\begin{equation} \small
\gkl[n] = \eta_{k,\l}\|\gkl[n-\Dkl]\|\left(\cos(\theta_{\tgkl, \hgkl}) \hgkl[n-\Dkl] + \sin(\theta_{\tgkl, \hgkl})\skl[n]  \right) + \sqrt{1-\eta_{k,\l}^2}\wkl[n]. \label{eqn:timeCorrApp}
\end{equation}
Since $\hfl[n]$ lies in the null-space of $\hgkl[n-\Dkl]$, we have
\begin{equation} 
|\gkl^T[n]\hfl[n]|^2  =  \Big|\eta_{k,\l}\|\gkl[n-\Dkl]\|\sin(\theta_{\tgkl, \hgkl})\skl[n] \hfl[n] + \sqrt{1-\eta_{k,\l}^2}\wkl[n]\hfl[n] \Big|^2 \nonumber .
\end{equation}
Using the triangle inequality, we get
\begin{align}
\begin{split}
|\gkl^T[n]&\hfl[n]|^2 \leq \left(\Big|\eta_{k,\l}\|\gkl[n-\Dkl]\|\sin(\theta_{\tgkl, \hgkl})\skl[n] \hfl[n]\Big| + \Big|\sqrt{1-\eta_{k,\l}^2}\wkl[n]\hfl[n] \Big|\right)^2 \\
& = \eta^2_{k,\l} \|\gkl[n-\Dkl]\|^2 \sin^2(\theta_{\tgkl, \hgkl}) \Big|\skl[n] \hfl[n]\Big|^2 + (1-\eta_{k,\l}^2) \|\wkl[n-\Dkl]\|^2\Big|\twkl[n] \hfl[n]\Big|^2 \\
& + 2 \eta_{k,\l}\|\gkl[n-\Dkl]\|\sin(\theta_{\tgkl, \hgkl})\sqrt{1-\eta_{k,\l}^2}\|\wkl[n-\Dkl]\|\Big|\skl[n] \hfl[n]\Big|.\Big|\wkl[n]\hfl[n]\Big|. \label{eqn:evalGF_5}
\end{split}
\end{align}
Evaluating $\bbE\{|\gkl^T[n]\hfl[n]|^2\}$ using the relations that $\Big|\skl[n] \hfl[n]\Big| \leq 1$ and $\Big|\wkl[n]\hfl[n]\Big| \leq 1$, 
\begin{align}
\begin{split}
\bbE\{|\gkl^T[n]\hfl[n]|^2\} & \leq \eta^2_{k,\l}\bbE\{\|\gkl[n-\Dkl]\|^2\} \bbE\{\sin^2(\theta_{\tgkl, \hgkl})\} \bbE\{\Big|\skl[n] \hfl[n]\Big|^2\} \\
& + (1-\eta_{k,\l}^2) \bbE\{\|\wkl[n-\Dkl]\|^2\}\bbE\{\Big|\twkl[n] \hfl[n]\Big|^2\} \\
& + 2 \eta_{k,\l}\sqrt{1-\eta_{k,\l}^2}\bbE\{\sin(\theta_{\tgkl, \hgkl})\}\bbE\{\|\gkl[n-\Dkl]\|\}\bbE\{\|\wkl[n-\Dkl]\|\}. \label{eqn:evalGF_7} 
\end{split}
\end{align}
Now, $\bbE\{\|\gkl[n-\Dkl]\|^2\} =  \bbE\{\|\wkl[n-\Dkl]\|^2\} = \Nt$. Further, since $\skl[n]$ and $\hfl[n]$ are both isotropically and independently distributed in the $\Nt - 1$ null-space of $\hgkl[n-\Dkl]$, $\Big|\skl[n] \hfl[n]\Big|^2$ is beta distributed as $\beta(1,\Nt-2)$ \cite{Jindal2006}. Similarly, $\twkl[n]$ and and $\hfl[n]$ are both isotropically and independently distributed in the $\Nt$ dimensions, implying that $\Big|\twkl[n] \hfl[n]\Big|^2$ is beta distributed as $\beta(1,\Nt-1)$. The mean of a beta distribution $\beta(a,b)$ is $\frac{a}{a+b}$. Since $\|\gkl[n-\Dkl]\|$ and $\|\wkl[n-\Dkl]\|$ are $\chi$ distributed with $2\Nt$ degrees of freedom, $\bbE\{\|\gkl[n-\Dkl]\|\} = \bbE\{\|\wkl[n-\Dkl]\|\} = \frac{\Gamma(\Nt + 1/2)}{\sqrt{2}\Gamma{\Nt}}$. Using these results, \eqref{eqn:evalGF_7} is rewritten as
\begin{eqnarray}
\bbE\{|\gkl^T[n]\hfl[n]|^2\}  & \leq & \eta^2_{k,\l} 2^\Bkl\beta\left(2^\Bkl,\frac{\Nt}{\Nt-1}\right)\frac{\Nt}{\Nt-1} + (1-\eta_{k,\l}^2) \nonumber \\
& &  + \eta_{k,\l}\sqrt{1-\eta_{k,\l}^2}\bbE\{\sin(\theta_{\tgkl, \hgkl})\} \left(\frac{\Gamma(\Nt + 1/2)}{\Gamma{\Nt}} \right)^2 . \label{eqn:evalGF_8}
\end{eqnarray}
The last term in the upper bound of $\bbE\left\{|\gkl^T[n]\hfl[n]|^2\right\}$ in \eqref{eqn:evalGF_8} includes the product of $\eta_{k,\l} (<1)$, for non-zero delay,  $\sqrt{1-\eta_{k,\l}^2} (<< 1)$ and $\bbE\{\sin(\theta_{\tgkl, \hgkl})\} (\leq 1)$, and hence, can be ignored. The lower-bound in \eqref{eqn:evalGF_8} can then be rewritten as
\begin{equation}
\bbE\{|\gkl^T[n]\hfl[n]|^2\}  \leq \log_2\left(1 +\rho_k\sum_{\substack{\l = 1 \\ \l \neq k}}^{K}\alpha_{k,\l} \left( 1-\eta_{k,\l}^2 +\eta^2_{k,\l} 2^\Bkl\beta\left(2^\Bkl,\frac{\Nt}{\Nt-1}\right)\frac{\Nt}{\Nt-1}\right)\right) . \label{eqn:UppBndT_int}
\end{equation}

\section{Proof of Theorem~\ref{app:OptSolnIntChan}}\label{app:OptSolnIntChan}
Applying the arithmetic mean - geometric mean inequality to the objective function in \eqref{eqn:OptProb-IntMinFinal}, we get 
\begin{equation}
\sum_{\substack{\l = 1 \\ \l \neq k}}^{K}\alpha_{k,\ell}\eta^2_{k,\ell}2^{-\frac{\Bkl}{\Nt-1}} \geq (K-1)\left(\prod_{\substack{\l = 1 \\\l \neq k}}^{K}\alpha_{k,\ell}\eta^2_{k,\ell}2^{-\frac{\Bkl}{\Nt-1}}\right)^{\frac{1}{K-1}}, \label{eqn:AM-GM-Bkl}
\end{equation}
where the equality holds for $\alpha_{k,\ell}\eta^2_{k,\ell}2^{-\frac{\Bkl}{\Nt-1}} = \alpha_{k,n}\eta^2_{k,n}2^{-\frac{B_{k,n}}{\Nt-1}},\l\neq n$. The minimum value of \eqref{eqn:AM-GM-Bkl} is 
\begin{equation}
(K-1)\alpha_{k,\ell}\eta^2_{k,\ell}2^{-\frac{\Bkl}{\Nt-1}} = (K-1)\left(2^{\frac{-\Bi}{\Nt-1}}\prod_{\substack{\l = 1 \\ \l \neq k}}^{K}\alpha_{k,\ell}\eta^2_{k,\ell}\right)^{\frac{1}{K-1}}.\label{eqn:GM-Bkl-2}
\end{equation}
Solving for $\Bkl$ gives us
\begin{equation}
\Bkl^* = \frac{\Bi}{K-1} + (\Nt-1)\log_2\left(\frac{\alpha_{k,\l}\eta^2_{k,\l} }{\prod_{\l = 1}^{K}(\alpha_{k,\l}\eta^2_{k,\l})^{\frac{1}{K-1}}}\right) .\label{eqn:Bkl}
\end{equation}
Note that \eqref{eqn:Bkl} can be negative when $\alpha_{k,\l}\eta^2_{k,\l}$ is sufficiently small. Since $\Bkl$ can only be non-negative, we set the $\Bkl = 0$ for all $\l$ which cause \eqref{eqn:Bkl} to become negative. This implies that we partition bits among only the effective set of interferers $\cK$ that will result in a non-negative value of $\Bkl$ for all $\l\in\cK$. 

\section{Proof of Theorem~\ref{thm:SolnLowSNR}}\label{app:SolnLowSNR}
By ignoring the terms in \eqref{eqn:OptProb3-LowSNR} that do not depend on $\Bk$, the new objective function is given by 
\begin{align}
\begin{split}
\log_2(e)\Gamma\left( \frac{\Nt}{\Nt-1}\right)2^{-\frac{\Bk}{\Nt-1}} +\log_2(e)\rho_k(K-1)\Gamma\left( \frac{2\Nt-1}{\Nt-1}\right)2^{-\frac{\Bt-\Bk}{(\Nt-1)(K-1)}}\prod_{\substack{\l = 1\\ \l \neq k}}^{K}(\alpha_{k,\l}\eta^2_{k,\l})^{\frac{1}{K-1}} . \label{eqn:OptProb3-LowSNR-2}
\end{split}
\end{align}
The expression in \eqref{eqn:OptProb3-LowSNR-2} can be further simplified by dividing the expression by the constant $\log_2(e)\Gamma\left( \frac{\Nt}{\Nt-1}\right)$ and by denoting $C_i = \rho_k\frac{(K-1)\Nt}{\Nt-1}2^{\frac{-\Bt}{(\Nt-1)(K-1)}}\prod_{\substack{\l = 1\\ \l \neq k}}^{K}(\alpha_{k,\l}\eta^2_{k,\l})^{\frac{1}{K-1}}$ as
\begin{align}
\begin{split}
2^{-\frac{\Bk}{\Nt-1}} +C_i2^{\frac{\Bk}{(\Nt-1)(K-1)}} . \label{eqn:OptProb3-LowSNR-3}
\end{split}
\end{align}
Using the arithmetic mean-geometric mean inequality, we get
\begin{equation}
2^{-\frac{\Bk}{\Nt-1}} +C_i2^{\frac{\Bk}{(\Nt-1)(K-1)}} \geq 2 \sqrt{2^{-\frac{\Bk}{\Nt-1}}\cdot C_i2^{\frac{\Bk}{(\Nt-1)(K-1)}} }.
\end{equation}
We can minimize \eqref{eqn:OptProb3-LowSNR-2} (and hence, \eqref{eqn:OptProb3-LowSNR}) by solving for $\Bk$ that satisfies the equality $2^{-\frac{\Bk}{\Nt-1}} = C_i2^{\frac{\Bk}{(\Nt-1)(K-1)}}$. Hence, the approximate mean loss in sum-rate at low SNR can be minimized by setting $\Bk$ as  
$$
\Bk^{\rm LS} = \frac{\Bt}{K} - \frac{(\Nt-1)(K-1)}{K} \log_2\left(\rho_k\frac{\Nt}{\Nt-1}\prod_{\substack{\l = 1\\ \l \neq k}}^{K}(\alpha_{k,\l}\eta^2_{k,\l})^{\frac{1}{K-1}}\right) .
$$

\section{Proof of Theorem~\ref{thm:SolnHighSNR}}\label{app:SolnHighSNR}
The approximate mean loss in sum-rate in \eqref{eqn:OptProb3-HighSNR} can be rewritten as 
\begin{align}
\begin{split}
\log_2(e)\Gamma\left( \frac{\Nt}{\Nt-1}\right)2^{-\frac{\Bk}{\Nt-1}} +&\log_2\left(1 + \frac{\Gamma\left( \frac{2\Nt-1}{\Nt-1}\right)2^{\frac{-\Bt}{(\Nt-1)(K-1)}}\prod_{\l = 1}^{K}(\alpha_{k,\l}\eta^2_{k,\l})^{\frac{1}{K-1}}}{\sum_{\substack{\l = 1 \\ \l \neq k}}^{K}\alpha_{k,\l} \left( 1-\eta_{k,\l}^2\right)}2^{\frac{\Bk}{(\Nt-1)(K-1)}}\right) \\
& - \log_2\left(\rho_k\sum_{\substack{\l = 1 \\ \l \neq k}}^{K}\alpha_{k,\l} \left( 1-\eta_{k,\l}^2\right)\right). \label{eqn:OptProb3-HighSNR-2}
\end{split}
\end{align}
The last term can be ignored in the optimization problem since it is independent of $\Bk$. Denoting the ratio before $2^{\frac{\Bk}{(\Nt-1)(K-1)}}$ inside the logarithm by $C_i$, the objective function in \eqref{eqn:OptProb3-HighSNR-2} can be simplified as
\begin{align}
\begin{split}
\log_2(e)\Gamma\left( \frac{\Nt}{\Nt-1}\right)2^{-\frac{\Bk}{\Nt-1}} +\log_2\left(1 +C_i2^{\frac{\Bk}{(\Nt-1)(K-1)}}\right) . \label{eqn:OptProb3-HighSNR-3}
\end{split}
\end{align}
Note that $C_i$ will typically be large since the denominator of $\sum_{\substack{\l = 1 \\ \l \neq k}}^{K}\alpha_{k,\l} \left( 1-\eta_{k,\l}^2\right)$ is very small. For large $C_i$, $\log_2\left(1 +C_i2^{\frac{\Bk}{(\Nt-1) (K-1)}}\right)\approx \log_2(C_i) +\frac{\Bk}{(\Nt-1)(K-1)}$, which is a linear (and hence, convex) function in $\Bk$. Now, $2^{-\frac{\Bk}{\Nt-1}}$ is a convex functions in $\Bk$. As the sum of non-negatively weighted convex functions is convex, the expression in \eqref{eqn:OptProb3-HighSNR-3} is convex. Hence, there exists a global minimizer, $\Bk \in [0,\Bt]$. Finding the derivative of the objective function in \eqref{eqn:OptProb3-HighSNR-3} with respect to $\Bk$, we get 
\begin{align}
\begin{split}
\frac{-1}{\Nt-1}\Gamma\left( \frac{\Nt}{\Nt-1}\right)2^{-\frac{\Bk}{\Nt-1}} + \frac{C_i 2^{\frac{\Bk}{(\Nt-1)(K-1)}}}{1+C_i2^{\frac{\Bk}{(\Nt-1)(K-1)}}} \frac{1}{(\Nt-1)(K-1)}. \label{eqn:OptProb3-HighSNR-4}
\end{split}
\end{align}
In \eqref{eqn:OptProb3-HighSNR-4}, we can approximate $\frac{C_i 2^{\frac{\Bk}{(\Nt-1)(K-1)}}}{1+C_i2^{\frac{\Bk}{(\Nt-1)(K-1)}}} \approx 1$, for large $C_i$. Setting \eqref{eqn:OptProb3-HighSNR-4} to zero, we get, 
$$
\Bk^{\rm HS} = (\Nt -1) \log_2\left( (K-1)\Gamma\left(\frac{\Nt}{\Nt-1}\right)\right) .
$$

\bibliographystyle{IEEEtran}
\bibliography{MulticellChanInv_Delay}

\begin{table}[ht] 
\caption{Simulation parameters, based on the 3GPP LTE's urban microcell.} \vspace{-10pt}
\centering      
\begin{tabular}{| l | c |}  
\hline                       
\textbf{Parameter} & \textbf{Value} \\ [0.5ex] \hline                    
Number of interferers & 6  \\   \hline
Carrier frequency, $f_c$ & $1.9~GHz$ \\ \hline
Base station height & $12.5~m$ \\ \hline
Mobile terminal height  & $1.5~m$\\ \hline
Cell radius & $400~m$ \\ \hline
Transmit power, $E_s$ & $3~dB$ \\ \hline
Noise power, $N_o$ & $-144~dB$ \\ \hline
Normalized feedback and quantization delay & 1 \\ \hline
Normalized backhaul delay & 1 \\  [1ex]  \hline     
\end{tabular} 
\label{tab:simParams}   
\end{table}

\centering{\large{\textbf{Figures}}\vspace{-10pt}
\begin{figure} [h!]
  \centering
  \includegraphics[width=3.75in]{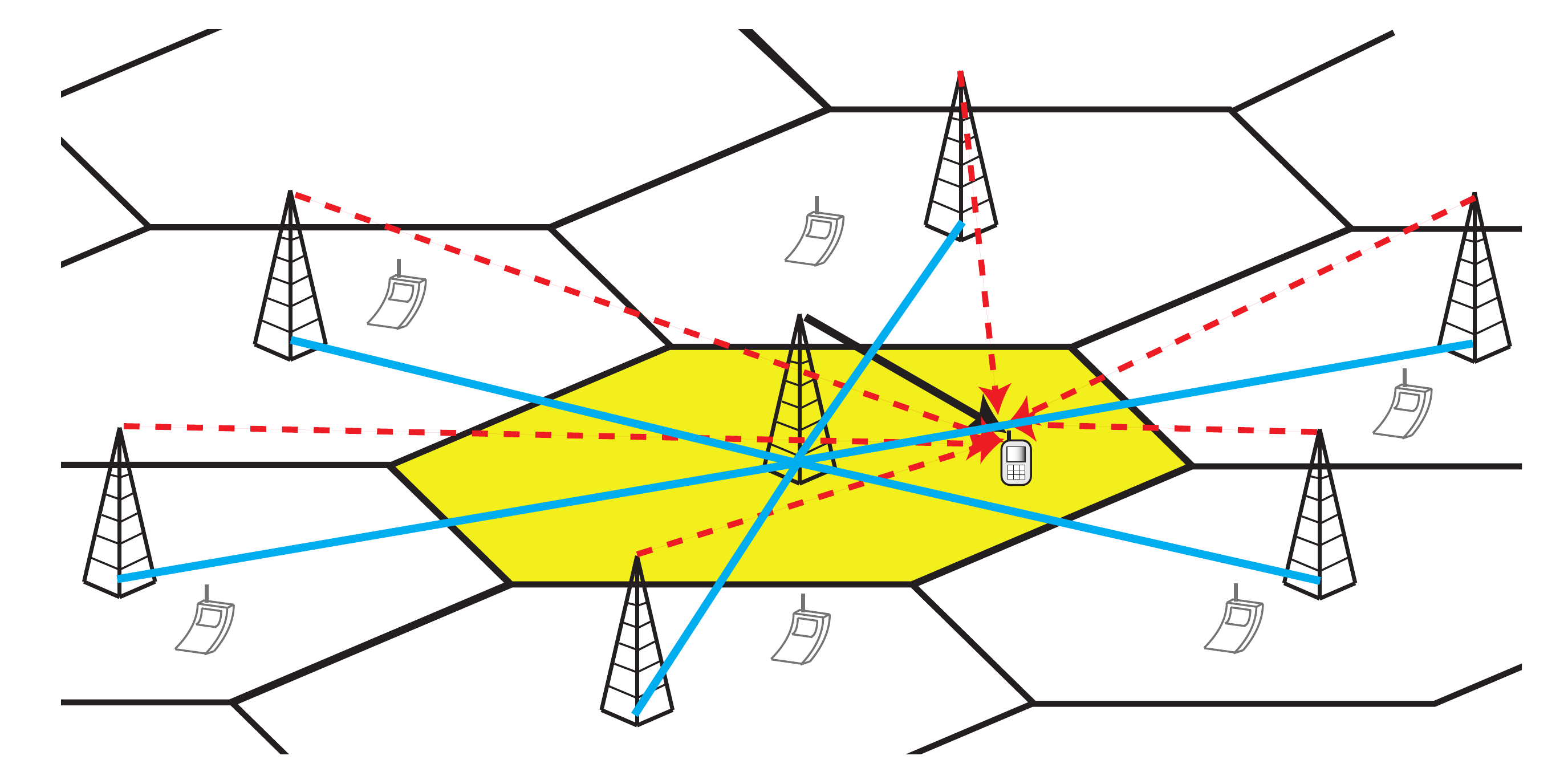}\vspace{-10pt}
  \caption{Multicell cooperative model described in Section~\ref{sec:SysModel}. The solid line with an arrow represents the   desired signal, while the dashed lines denote the interfering signals. The solid line between base stations represents the backhaul links connecting base stations.}\label{fig:SystemModel}\vspace{-10pt}
\end{figure}

\begin{figure} [h!]
  \centering
  \includegraphics[width=6.25in]{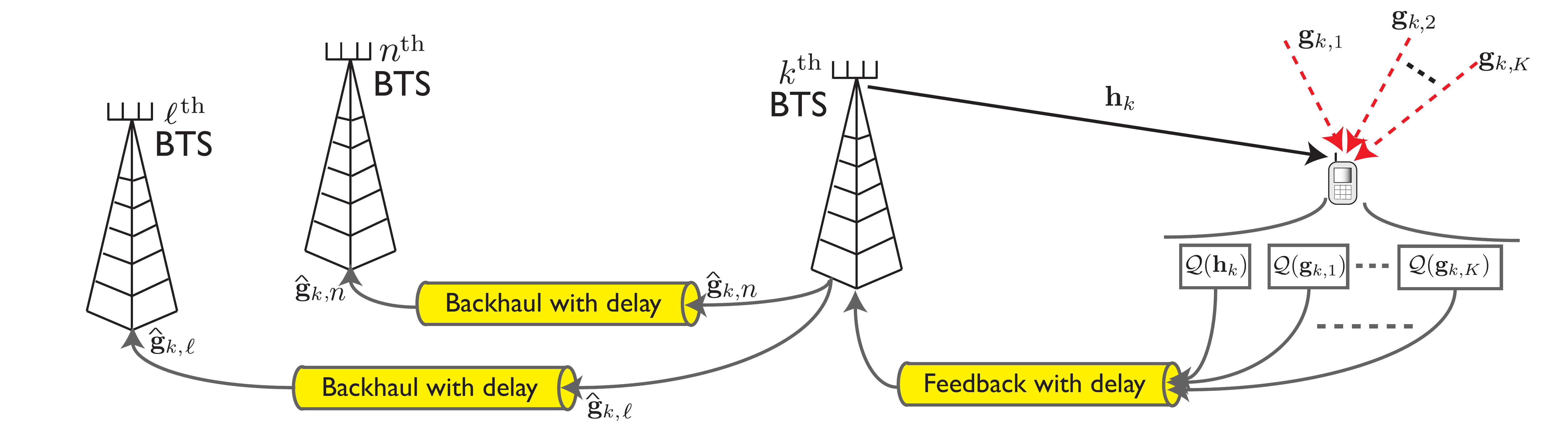}\vspace{-10pt}
  \caption{The limited feedback model, described in Section~\ref{sec:LFModel}, to feedback quantized CSI of the desired and interfering channels. The quantizing operation is denoted by $\cQ$.}\label{fig:LFModel}\vspace{-10pt}
\end{figure}

\begin{figure} [h!]
  \centering
  \includegraphics[width=4.25in]{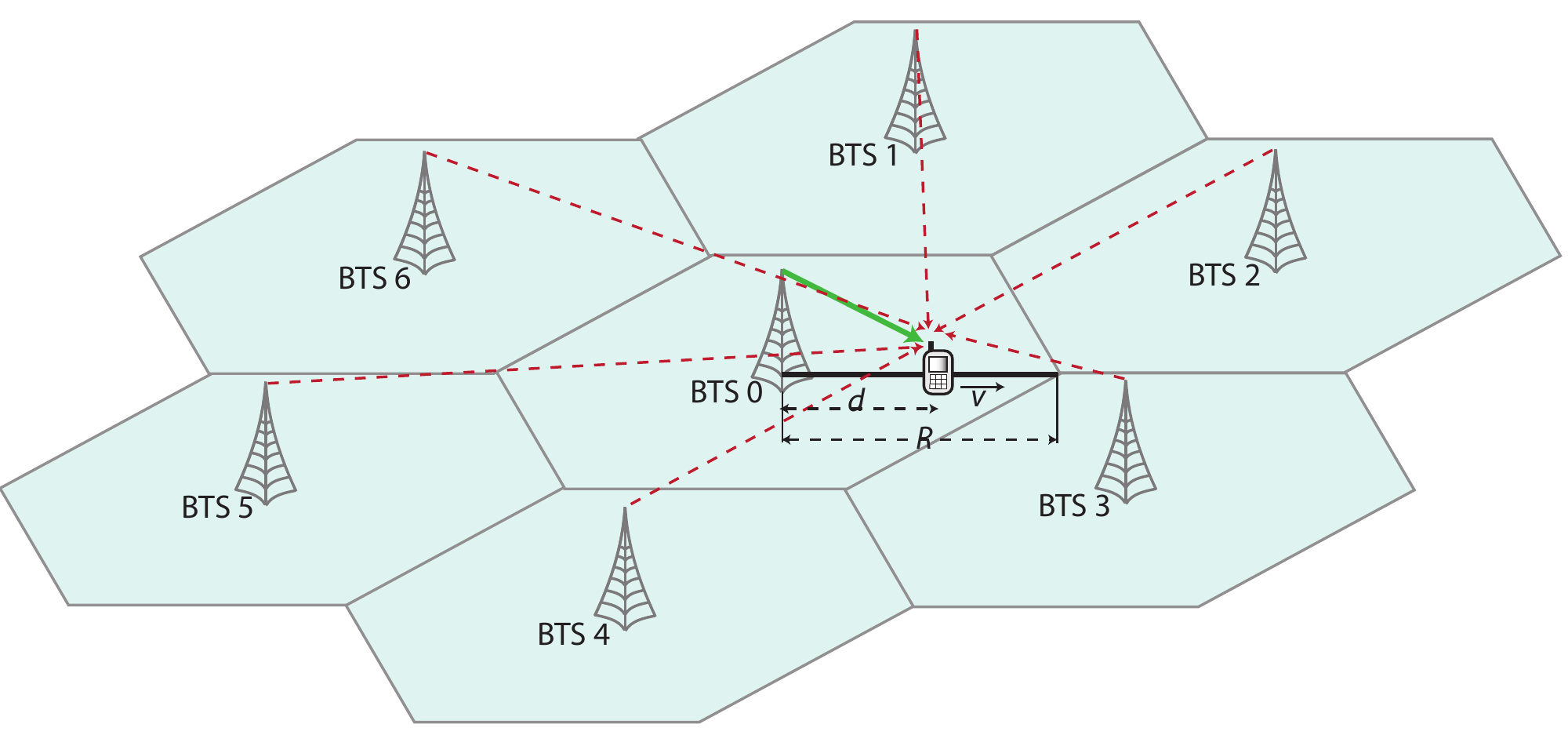}\vspace{-10pt}
  \caption{The simulation setup in Section~\ref{sec:Results}. The base stations are numbered in a clockwise fashion.}\label{fig:SimsModel}\vspace{-10pt}
\end{figure}

\begin{figure} [h!]
  \centering
  \includegraphics[width=3.25in]{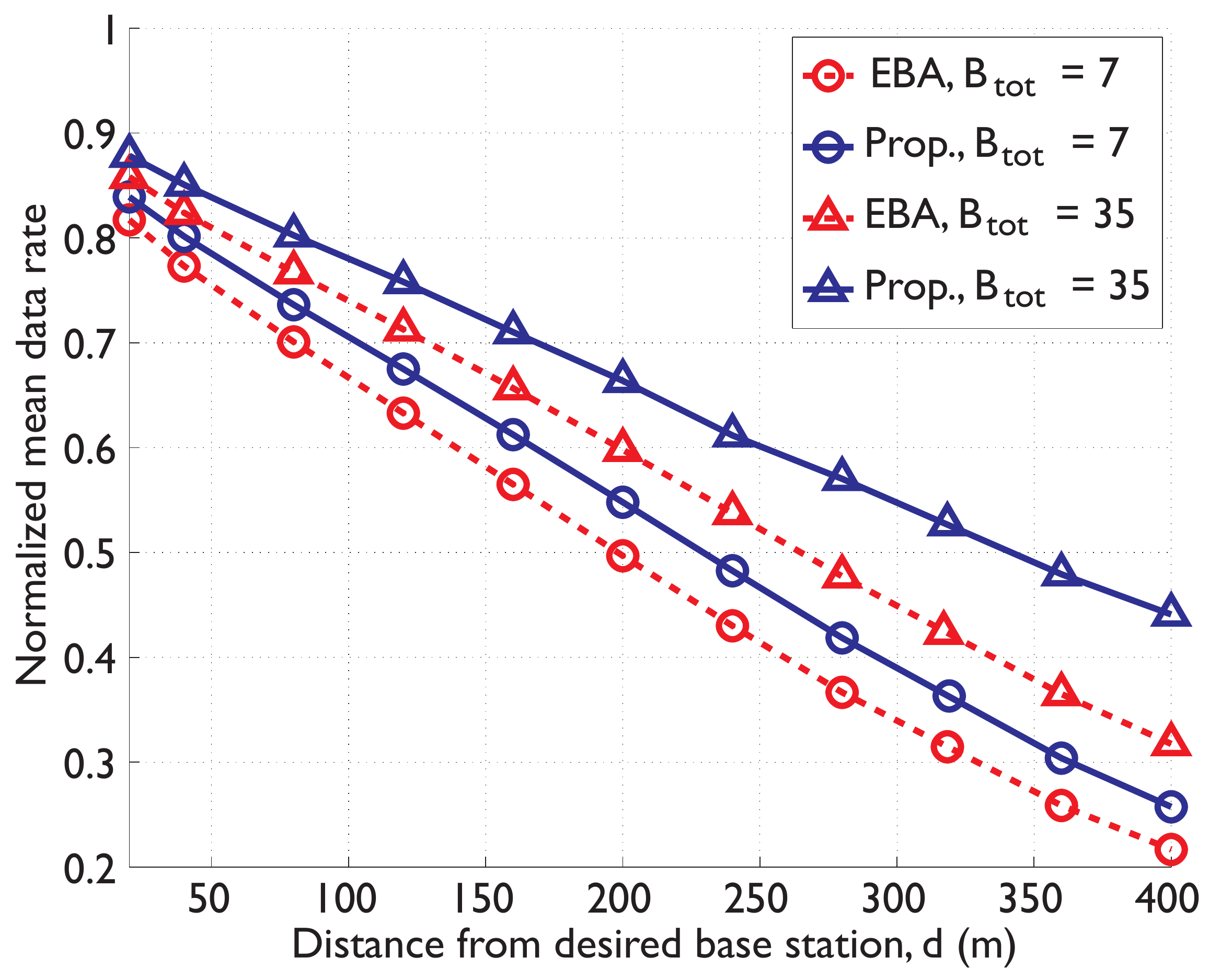}\vspace{-10pt}
  \caption{Comparison of the normalized data-rate as a function of the distance, $d$, of the user from the desired base station for $\Bt = 7, 35$ and $v=10~mph$. }\label{fig:dataRateVsDist}\vspace{-10pt}
\end{figure}

\begin{figure}[h!]
  \centering
  \subfigure[High SNR partitioning, for $E_s = 3~dB$.]{\label{fig:bitPartBtot35HighSNR}\includegraphics[width=3.0in]{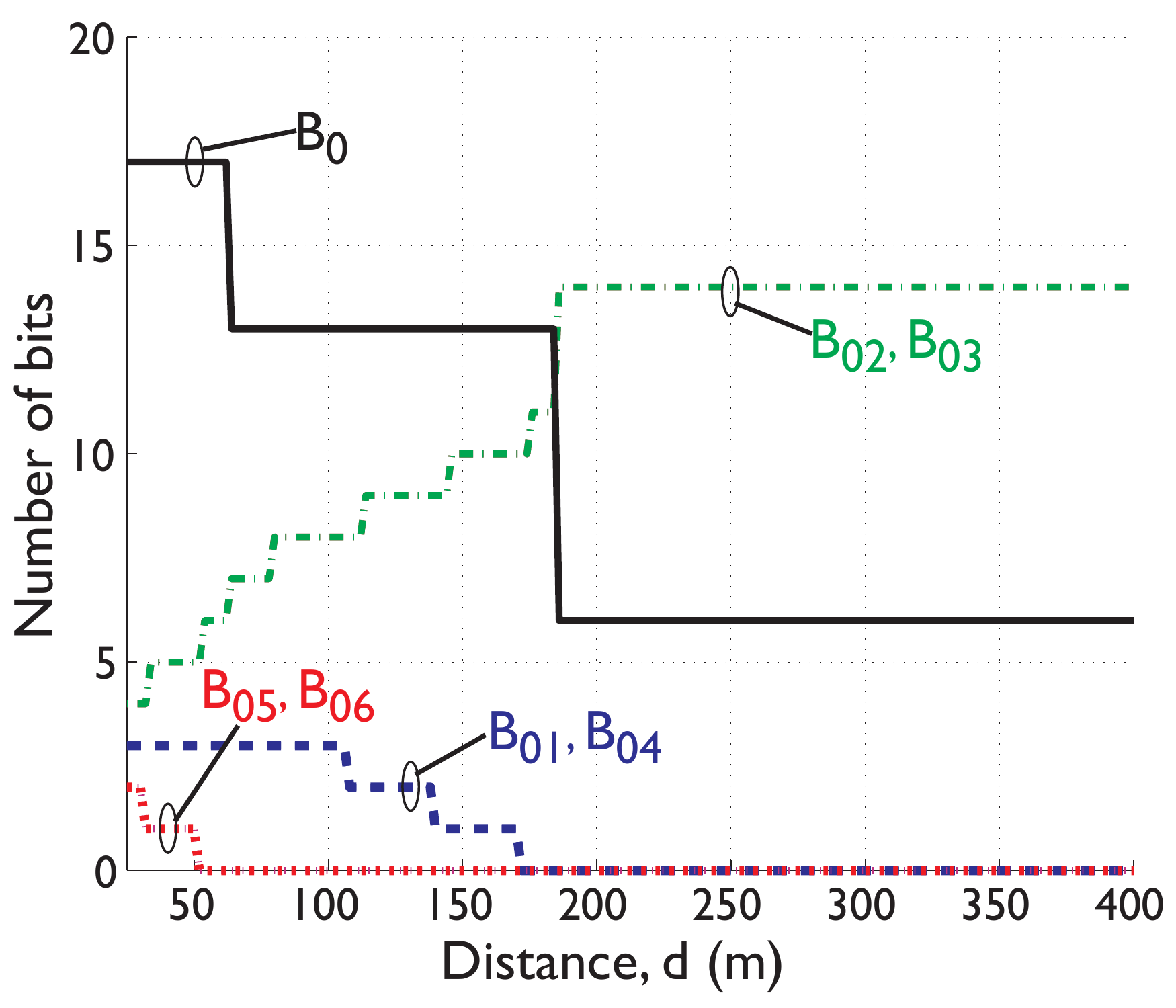}}                \vspace{-5pt}
  \subfigure[Low SNR partitioning, for $E_s = -3~dB$.]{\label{fig:bitPartBtot35LowSNR}\includegraphics[width=3.0in]{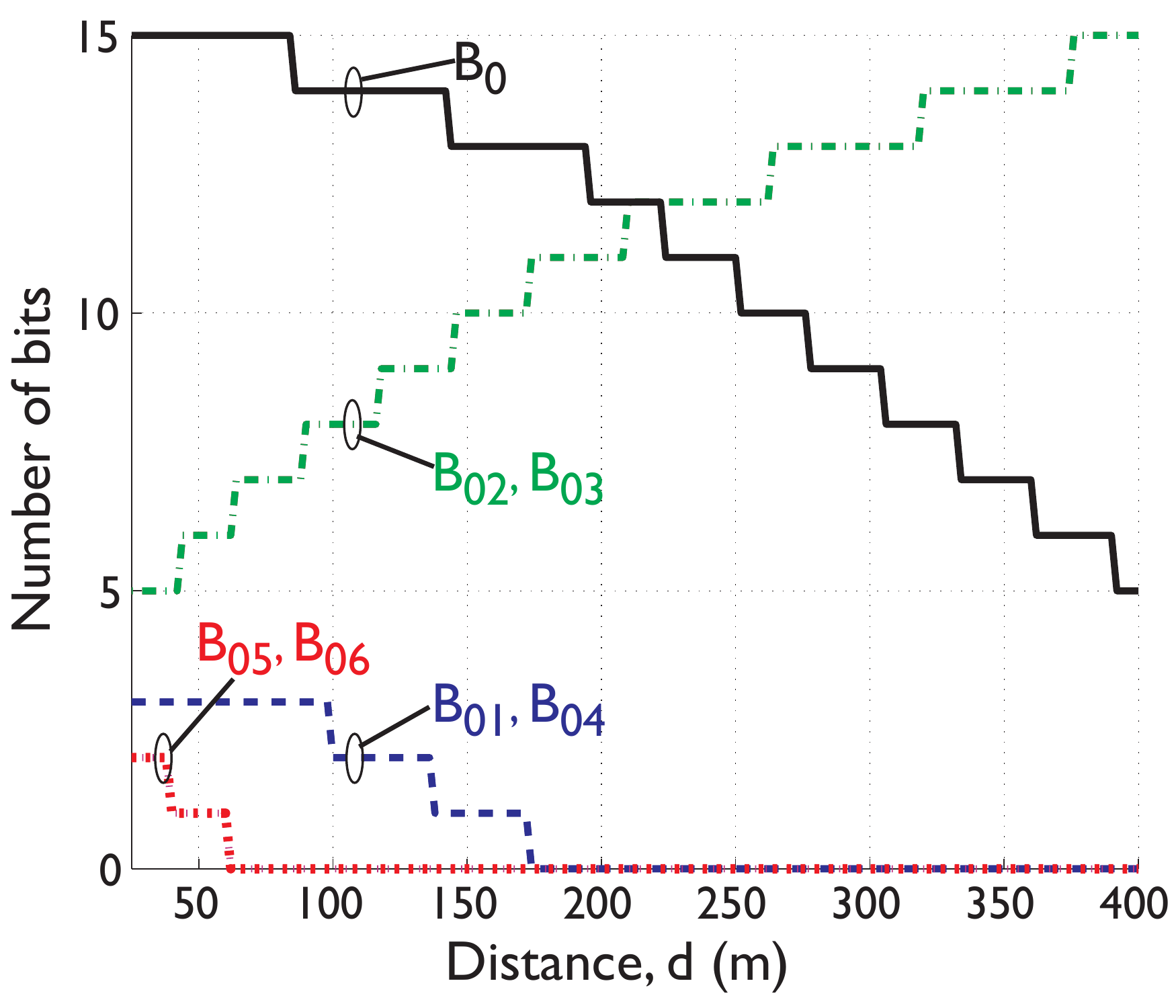}}\vspace{-5pt}
  \caption{Adaptive feedback-bit partitioning as a function of the distance, $d$, the user from the desired base station for $\Bt = 35$ and $v=10~mph$ for (a) $E_s = 3 dB$ and (b) $E_s = -3 dB$ .}
  \label{fig:bitPart}
\end{figure}

\begin{figure} [h!]
  \centering
  \includegraphics[width=4.75in]{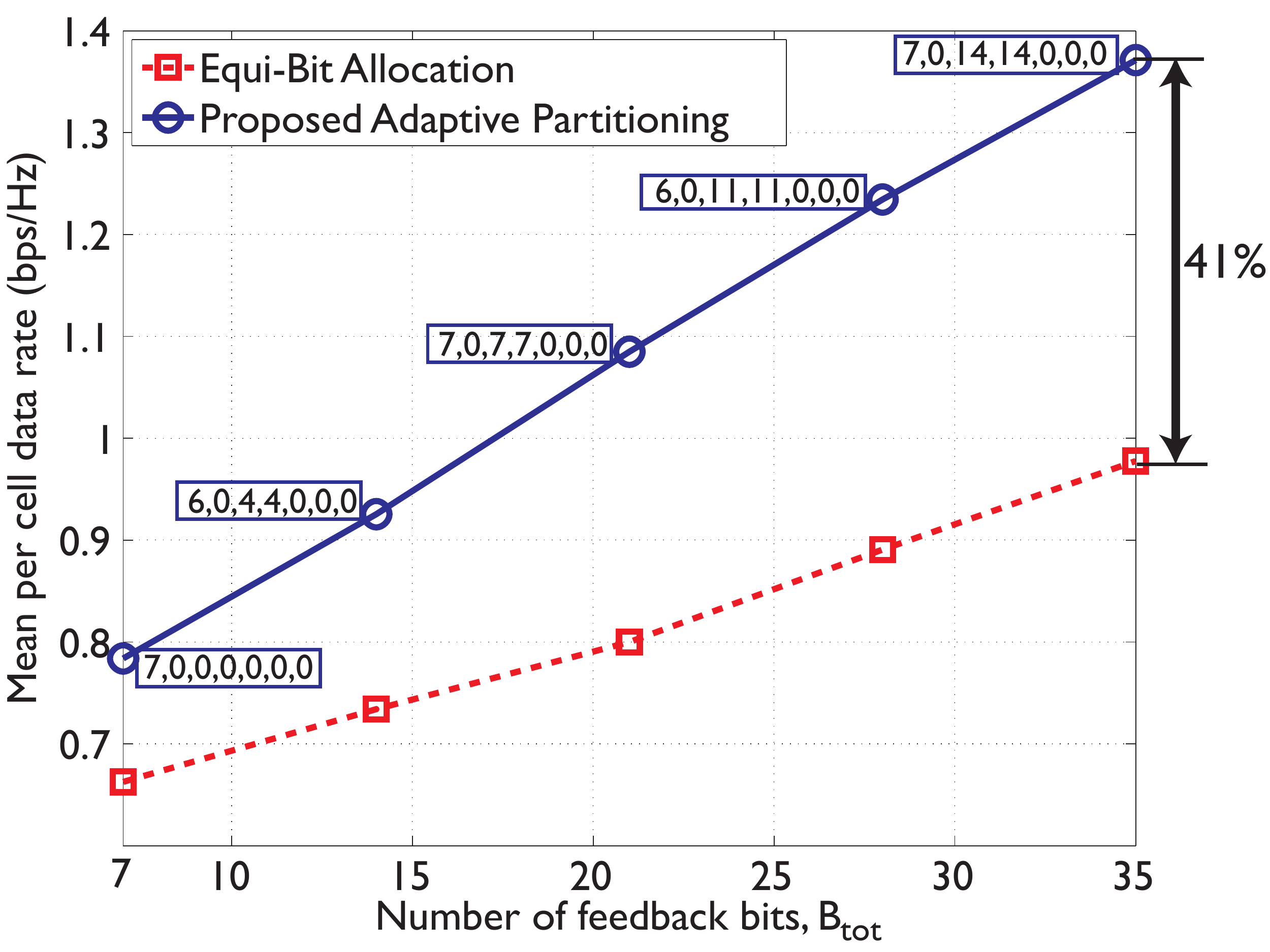}
  \caption{Comparison of the mean data-rate at the cell-edge for different values of $\Bt$. Bit assignments are shown corresponding to each $\Bt$.}\label{fig:dataRateVsBtot}
\end{figure}

\begin{figure} [h!]
  \centering
  \includegraphics[width=4.75in]{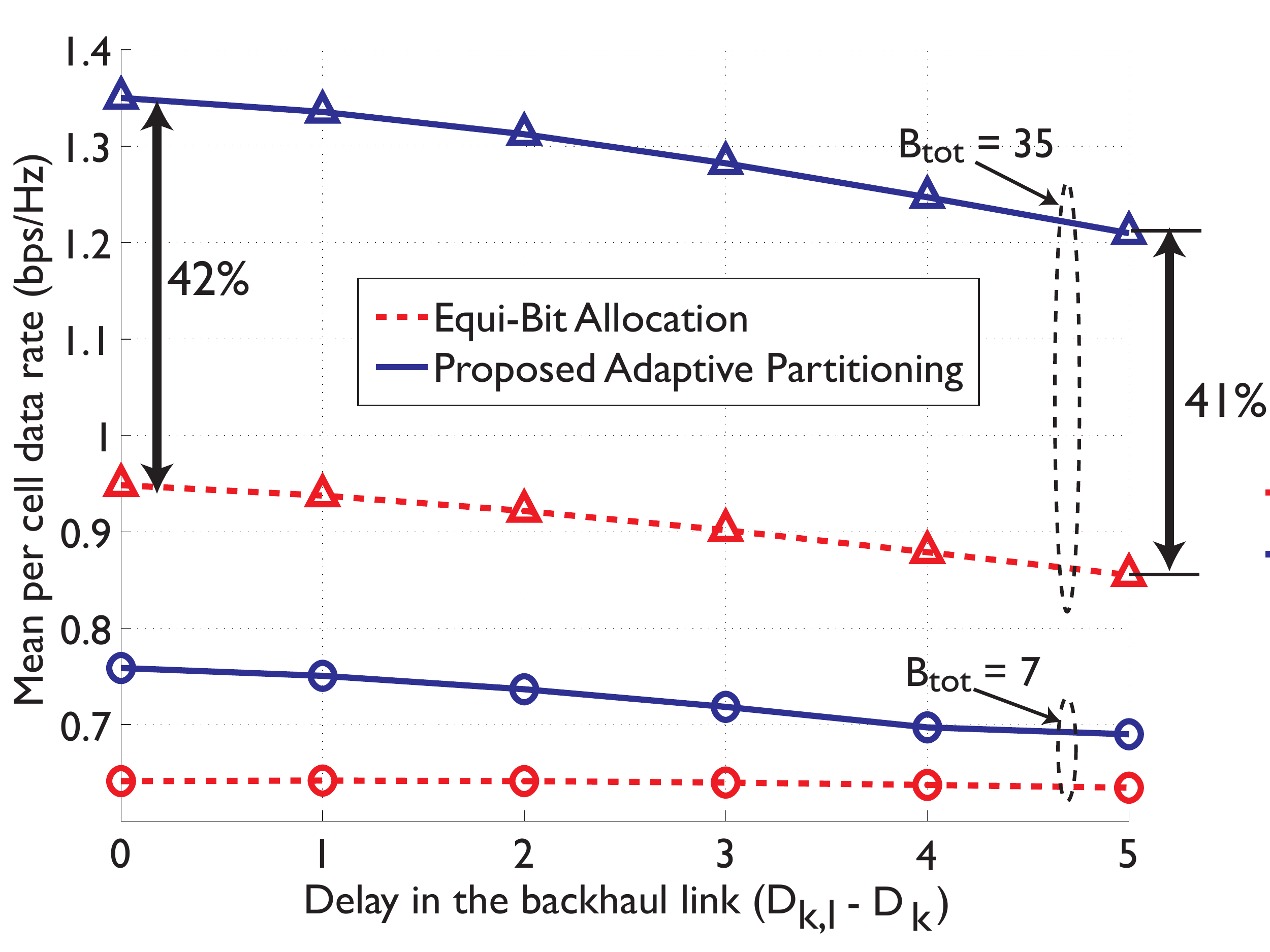}
  \caption{Comparison of the mean data-rate at the cell-edge for different values of delay in the backhaul link.}\label{fig:delay}
\end{figure}

\end{document}